\documentclass[final]{article}
\usepackage{ijcai19}

\newif\iftechrep
\techreptrue

\iftechrep
\fi

\usepackage{times}
\usepackage{soul}
\usepackage{url}
\usepackage[hidelinks]{hyperref}
\usepackage[utf8]{inputenc}
\usepackage[small]{caption}
\usepackage{graphicx}
\usepackage{amsmath,amssymb,amsfonts}

\usepackage{booktabs}
\usepackage{algorithm}
\usepackage[noend]{algorithmic}
\usepackage[inline]{enumitem}
\usepackage{microtype}
\usepackage{array}
\usepackage{xspace}

\makeatletter
\def\BState{\State\hskip-\ALG@thistlm}
\makeatother

\usepackage[thmmarks,amsmath,thref]{ntheorem}

\theoremseparator{.}
\theorembodyfont{\itshape}
\newtheorem{theorem}{Theorem}
\newtheorem{proposition}{Proposition}
\newtheorem{lemma}{Lemma}
\newtheorem{corollary}{Corollary}
\newtheorem{claim}{Claim}

\theoremseparator{.}
\theorembodyfont{\upshape}
\theoremsymbol{\mbox{$\triangleleft$}}
\newtheorem{definition}{Definition}

\theorembodyfont{\upshape}%\small
\theoremsymbol{\mbox{$\triangleleft$}}
\newtheorem{example}{Example}
\newtheorem{remark}{Remark}

\theoremstyle{nonumberplain}
\theoremheaderfont{\scshape}
\theorembodyfont{\upshape}%\small
\theoremsymbol{\mbox{$\square$}}
\newtheorem{proof}{Proof}

\usepackage{thm-restate} % needs to come after ntheorem

\usepackage{tikz}
\usepackage{pgfplots}
\usetikzlibrary{patterns}
\usetikzlibrary{backgrounds}
\usetikzlibrary{calc}
\usetikzlibrary{positioning,fit}
\usepgfmodule{plot}

\usepackage{courier}
\usepackage{listings}
\usepackage{misc0} % check for style conflicts with packages loaded here
\newcommand{\ontop}{\textit{Ontop}\xspace}
\newcommand{\ontoprov}{\textit{OntoProv}\xspace}

\renewcommand{\O}{\mathcal{O}}
\newcommand{\M}{\mathcal{M}}
\renewcommand{\S}{\mathcal{S}}

\newcommand{\D}{\mathcal{D}}

\newcommand{\PR}{\mathit{PR}}

\newcommand{\dataInstance}{\ensuremath{\Dmc}\xspace}
\newcommand{\assertionName}{\ensuremath{E}\xspace}
\newcommand{\obdaSpecification}{\ensuremath{\Pmc}\xspace}

\newcommand{\database}{\ensuremath{\Dmc}\xspace}
\newcommand{\mapping}{\ensuremath{\Mmc}\xspace}
\newcommand{\dataSchema}{\ensuremath{\Smc}\xspace}
\newcommand{\p}[2]{\ensuremath{{\sf prov}_{#1}({#2})}\xspace}
\newcommand{\tr}[2]{\ensuremath{{\sf Tr}({#1},{#2})}\xspace}
\newcommand{\map}{\ensuremath{{\pi}}\xspace}
\renewcommand{\rule}{\ensuremath{{\rho}}\xspace}
\newcommand{\pair}[2]{\ensuremath{({#1},{#2})}}
\newcommand{\polynomial}{\ensuremath{p}\xspace}
\newcommand{\semiring}{\ensuremath{\mathbb{K}}\xspace}
\newcommand{\semiringVariable}{\ensuremath{p}\xspace}
\newcommand{\semiringVariables}{\ensuremath{\NV}\xspace}

\newcommand{\sfm}{\ensuremath{{\sf m}}\xspace}
\newcommand{\pminus}[2]{\ensuremath{#1_{| #2}}\xspace}
\newcommand{\pplus}[2]{\ensuremath{{#1}\otimes {#2}}\xspace}
\newcommand{\pplusidem}[2]{\ensuremath{{#1}^{#2}}\xspace}

\newcommand{\perfectRef}{\ensuremath{{\sf PerfectRef}}\xspace}
\newcommand{\perfectRefIdem}{\ensuremath{{\sf PerfectRef}^{\star}}\xspace}
\newcommand{\computeProv}{\ensuremath{{\sf ComputeProv}}\xspace}
\newcommand{\marked}{marked\xspace}

\renewcommand{\oplus}{\ensuremath{+}\xspace}
\renewcommand{\otimes}{\ensuremath{\times}\xspace}

\begin{document}

\iftechrep
\title{Enriching Ontology-based Data Access with Provenance\\(Extended Version)}
\else
\title{Enriching Ontology-based Data Access  with Provenance}
\fi
\author{Diego Calvanese$^1$\and Davide Lanti$^1$\and Ana Ozaki$^1$\and Rafael
 Pe\~naloza$^2$\And Guohui Xiao$^1$
 \affiliations
 $^1$KRDB Research Centre,
 Free University of Bozen-Bolzano, Italy \\
 $^2$University of Milano-Bicocca, Italy}

\maketitle

\begin{abstract}
  Ontology-based data access (OBDA) is a popular paradigm for querying
  heterogeneous data sources by connecting them through mappings to an
  ontology.  In OBDA, it is often difficult to reconstruct why a tuple occurs in
  the answer of a query.  We address this challenge by enriching OBDA with
  provenance semirings, taking inspiration from database theory.  In
  particular, we investigate the problems of \textit{(i)}~deciding whether a
  provenance annotated OBDA instance entails a provenance annotated conjunctive
  query, and \textit{(ii)}~computing a polynomial representing the provenance
  of a query entailed by a provenance annotated OBDA instance.  Differently
  from pure databases, in our case these polynomials may be infinite.  To
  regain finiteness, we consider
  % make the natural assumption that operations in the polynomials are
  idempotent semirings, and study the complexity in the case of \DLLITE
  ontologies.  We implement Task~\textit{(ii)} in a state-of-the-art OBDA
  system and show the practical feasibility of the approach through an
  extensive evaluation against two popular benchmarks.
\end{abstract}

% !TEX root =  paper.tex

\section{Introduction}

Ontology-based data access (OBDA) \cite{XCKL*18} is by now a popular paradigm
which has been developed in recent years to overcome the difficulties in
accessing and integrating legacy data sources.  In OBDA, users are provided
with a high-level conceptual view of the data in the form of an ontology that
encodes relevant domain knowledge.  The concepts and roles of the ontology are
associated via declarative mappings to SQL queries over the underlying
relational data sources.  Hence, user queries formulated over the ontology can
be automatically rewritten, taking into account both ontology axioms and
mappings, into SQL queries over the sources.

When issuing a query,
% one may be interested not only in its result,
in many settings it is crucial to know not only its result
but also \emph{how} it was produced, \emph{how many} different ways there are
to derive it, or how \emph{dependent} it is on certain parts of the data
% , among many other questions
%\cite{DBLP:journals/sigmod/Senellart17,DBLP:journals/ws/ZimmermannLPS12,DBLP:conf/semweb/BunemanK10}.
\cite{DBLP:journals/sigmod/Senellart17,DBLP:journals/ws/ZimmermannLPS12,DBLP:conf/semweb/BunemanK10}.
To address these issues, which are of importance already
% in the standard setting where data is managed by a
for plain relational database management systems (RDBMSs), \emph{provenance
 semirings} \cite{Green07-provenance-seminal,GreenT17} were introduced as an
abstract tool to record and track provenance information; that is, to keep
track of the specific database tuples that are responsible for deriving an
answer tuple, and of additional information associated to them.
In OBDA, determining provenance is made even more challenging by the fact that
answers are affected by implicit consequences derived through ontology axioms,
and by the use of mappings.  Such elements come indirectly into play in query
rewriting, hence provenance information must be reconstructed from the
rewritten queries used in the answering process \cite{BoCR08b}.

In this work, we start from the semiring approach introduced for RDBMSs, and
extend it to the full-fledged OBDA setting.  To do so, we assume that not only
database tuples are annotated with a label representing provenance information
(e.g., the data source or the relation in which the tuple is stored), but also
mappings and ontology axioms.  Then, our task is to derive which combinations
of these labels lead to the answer of a query.
Such information is expressed through a \emph{provenance polynomial}, as
illustrated in the following example.
\begin{example}\label{exa:mapping}
  Let $\ex{Mayors}[\ex{Person},\ex{City}]$ be a database relation with the
  tuples $(\ex{Renier}, \ex{Venice})$ and $(\ex{Brugnaro},\ex{Venice})$,
  annotated with (sources) $p$ and $q$, respectively. Assume two mappings
  $\ex{City}(Y) \leftarrow \ex{Mayors}(X,Y)$ and
  $\ex{headGov}(X,Y) \leftarrow \ex{Mayors}(X,Y)$, annotated with $m$ and~$n$,
  respectively.
  % The set of DL assertions is \emph{induced} by the mappings.
  % such setting is obtained by evaluating the queries on the RHS of the
  % mappings the database relation, and using the obtained answers to
  % instantiate the corresponding variables on the LHSs. In our case,
  The mappings and the database \emph{induce}
  \begin{enumerate*}[label=\textit{(\roman*)}]
  \item two times the DL assertion $\ex{City}(\ex{Venice})$, one annotated with
    $p \times m$ and one with $q \times m$,
  \item the DL assertion $\ex{headGov}(\ex{Renier},\ex{Venice})$, annotated
    with $p \times n$, and
  \item the assertion $\ex{headGov}(\ex{Brugnaro},\ex{Venice})$, annotated with
    $q \times n$.
  \end{enumerate*}

  Now consider the inclusion $\exists \ex{headGov}\sqsubseteq \ex{Mayor}$
  annotated with $s$.  The answer $\ex{true}$ to the Boolean conjunctive query
  $\exists x. (\ex{Mayor}(x))$ can be derived using this inclusion and any of
  the last two DL assertions.  This information can be expressed through the
  provenance polynomial $((p \times n) \oplus (q \times n)) \otimes s$.
\end{example}
%
% \begin{example}
%   \label{exa:intro}
%   Consider the DL assertions $\ex{City}(\ex{Venice})$,
%   $\ex{headGov}(\ex{Renier}, \ex{Venice})$, and
%   $\ex{headGov}(\ex{Brugnaro},\ex{Venice})$, annotated with the sources $p$,
%   $q$, and $r$ respectively. The %answer to the
%   query $\exists xy. (\ex{headGov}(x,y)\wedge \ex{City}(y))$ can be answered
%   using any of the last two assertions %(the role assertions)
%   together with the
%   first assertion. This is expressed by the provenance polynomial
%   $(p\oplus q) \otimes r$.
% \end{example} \ano{say that we also annotate inclusions?}
%
In our OBDA setting,
% we do not only label the facts in the data sources, but also the
concept and role inclusions of the ontology affect query results, as
illustrated in Example~\ref{exa:mapping}.  By annotating the inclusions and the
mappings, in addition to the tuples, we can distinguish which inclusions and
mappings were involved in the derivation of a query result.  This differs from
the approach proposed for attributed \DLLITE \cite{bourgauxozaki}, where the
inclusions are used to express constraints on the provenance information.

% These labels
% are interpreted and propagated according to a semiring semantics,
% as we illustrate in Example~\ref{exa:introinclusion}.
% \begin{example}
%   \label{exa:introinclusion}
%   Consider %the annotated DL assertions of Example~\ref{exa:intro} and
%   the inclusion $\exists \ex{headGov}\sqsubseteq \ex{Mayor}$ annotated with
%   the source $s$. Together with the annotated assertions from
%   Ex.~\ref{exa:intro}, the answer to $\exists x. (\ex{Mayor}(x))$ can be
%   derived using any of the last two assertions and the inclusion. In terms of
%   source annotations, this information can be expressed through the
%   provenance polynomial $(p\oplus q) \otimes s$.
% \end{example}

We investigate the problems of
\begin{enumerate*}[label=\textit{(\roman*)}]
\item deciding whether a provenance annotated OBDA instance entails a
  provenance annotated conjunctive query (CQ), and
\item computing a provenance polynomial of a CQ entailed by a provenance
  annotated OBDA instance.
\end{enumerate*}
Differently from plain databases, in our case these polynomials may be infinite.
To regain finiteness, we consider
% make the natural assumption that operations in the polynomials are
idempotent semirings, and study the complexity for % in the case of
\DLLITE ontologies \cite{dl-lite}.  We implement task~\textit{(ii)} in the
state-of-the-art OBDA system \ontop \cite{CCKK*17}, and show the practical
feasibility of our approach through a detailed
evaluation against two popular benchmarks.

\iftechrep
This article is an extended version of \cite{provenance-IJCAI-2019},
with selected proofs and additional information provided in an appendix.
\else
An extended version of this work
% with a link to our \ontoprov system
is available as a technical report~\cite{ourtechreport-arXiv-2019}.
% as well as our \ontoprov system, are provided at:
% \url{https://tinyurl.com/ijcai19-obda-prov}.
\fi

% !TEX root =  paper.tex
\section{Basic Definitions}
\label{sec:prelim}

We represent the provenance information
via a \emph{positive algebra provenance semiring} (or \emph{provenance
 semiring} for short),
originally introduced for databases \cite{Green07-provenance-seminal}.
Given a countably infinite set $\semiringVariables$ of \emph{variables}, the provenance semiring is the algebra
$\semiring=(  \mathbb{N}[\semiringVariables], \oplus ,\otimes,0,1)$,
% $\semiring=(  \mathbb{N}[X], \oplus ,\otimes,0,1)$,
where $\mathbb{N}[\semiringVariables]$ denotes the space of polynomials with
coefficients in $\mathbb{N}$ and variables in $\semiringVariables$, the product
$\otimes$ and the addition $\oplus$ are two commutative and associative binary
operators over $\mathbb{N}[\semiringVariables]$, and~$\otimes$ distributes
over~$\oplus$.  
A \emph{monomial} from
% the provenance semiring
\semiring is a finite product of variables in $\semiringVariables$.
\NM and \NPr denote the sets of all monomials from
% the provenance semiring
\semiring, and of all finite sums of monomials in \NM, respectively;
%polynomials over the variables $\semiringVariables$
%
%that are of \emph{expanded} form;
i.e., \NPr contains only polynomials of the form $\sum_{1\leq i \leq n}\prod_{1\leq j_i \leq m_i}a_{i,j_i}$,
with $a_{i,j_i}\in \semiringVariables$, and $n,m_i>0$. %~\footnote{S
Since all coefficients
are in $\mathbb{N}$,  they disappear in this \emph{expanded form}; e.g., $2a$ is $a+a$. %}.
A polynomial in expanded form is a finite sum of  monomials, each formed by a finite
product of variables. By distributivity, every
polynomial can be equivalently rewritten in expanded form; however, the
expanded form of a polynomial
may become exponentially larger.
%\ano{We do not need this assumption. Consider the expanded form of (a1+b1)(a2+b2)...(an+bn)}
%\rpn{I agree with Ana}
% \added[DC]{ (
%(under the usual assumption that the
% coefficients are represented in binary). }.
By our definitions, $\NV\subseteq \NM\subseteq \NPr$.
%\ano{maybe add an example}
%For instance, the polynomial $(p\oplus q) \otimes r$ from Example~\ref{exa:intro} is equivalent to
%$(p\otimes r)\oplus (q\otimes r)$.

\mypar{Annotated OBDA}
The provenance information of %for? or of?
each axiom in an ontology, each mapping, and each tuple in a data source,
%\Omc and each formula in \Mmc
is stored as an annotation.
%\dc{I would move the following par before where we introduce annotated OBDA
% instance}%
%
For this paper, we consider the standard OBDA setting with ontologies written
in % the DL
{\DLLITE} \cite{dl-lite}, standard  relational databases as data sources,
and mappings given by %specified in terms of
GAV rules.
Consider three mutually disjoint countable sets of
\emph{concept names} \NC, \emph{role names} \NR, and \emph{individual names} \NI.  
Assume that these sets are also disjoint from \NV.
 %~ \DLLITE \emph{role} and \emph{concept assertions}
%~ are expressions of the form
%~ $R(a,b)$ and $A(a)$, respectively, where  $R\in\NR$, $a,b\in\NI$, and
%~ %$A(a)$, with
%~ $A\in \NC$.
%, res;
%and a \DLLITE \emph{role assertion} is an expression of the form
%.
\DLLITE \emph{role} and \emph{concept inclusions} are expressions of the form
$S\sqsubseteq T$ and $B\sqsubseteq C$, respectively, where $S$, $T$ are role
expressions and $B$, $C$ are concept expressions built through the grammar
rules
\[
  S::=R\mid R^-,\ 
  T::= S\mid \neg S, \ 
  B::= A\mid \exists S, \ 
  C::= B\mid \neg B,
\]
with $R\in \NR$ and $A\in \NC$.
A \DLLITE \emph{axiom} is a \DLLITE role or  concept inclusion.
An \emph{  annotated} \DLLITE \emph{ontology} %, called \DLLITEp,
is a finite set of \emph{annotated axioms} of the form
$\pair{\alpha}{\polynomial}$, where $\alpha$ is a \DLLITE axiom
and $\polynomial\in\NM$. %  is a monomial  from the provenance semiring. % \semiring.

%We now introduce   annotated  data instances and mappings.
A \emph{schema} \dataSchema is a finite set of predicate symbols
disjoint from $\NC\cup\NR$ %. We denote by
  with ${\sf ar}(P)$   the arity of $P\in\dataSchema$.
 An  \emph{  annotated data instance} \dataInstance over \dataSchema maps every
$P\in \dataSchema$  to a finite subset $P^\dataInstance$ of
$\NI^{{\sf ar}(P)}\times \semiringVariables$.
 An \emph{  annotated mapping} %\Mmc
is a finite set of \emph{annotated rules} %of the form
$\pair{\rule}{\polynomial}$,
where    $\rule$ is a (GAV) rule and $\semiringVariable\in\semiringVariables$.
%
%We may treat tuples as sets when the order is irrelevant.
%
A rule $\rho$ is
of the form $\assertionName(\vec{x})\leftarrow\varphi(\vec{x},\vec{y},\vec{z})$,
with $\assertionName\in\NC\cup\NR$ and $\varphi(\vec{x}, \vec{y},\vec{z})$
a conjunction of atoms $P(\vec{t},t)$,
with $P\in\dataSchema$, $\vec{t}$   an ${\sf ar}(P)$-tuple
of terms in $\vec{x}\cup \vec{y}$, and
$t\in\vec{z}$.
% The definition of $\rule$ as a conjunction of atoms follows the idea that
% semirings capture the provenance of positive queries
% \cite{Green07-provenance-seminal}.
We restrict $\varphi$ to a conjunction of atoms for simplicity of our
theoretical development, also in line with the idea that semirings capture the
provenance of positive queries \cite{Green07-provenance-seminal}.  See
Sec.~\ref{sec:evaluation} for handling arbitrary OBDA mappings in our
implementation.

An \emph{annotated OBDA specification} $\Pmc$ is a triple
$(\Omc,\Mmc,\dataSchema)$, where \Omc is an ontology with annotated axioms,
\dataSchema is a data source schema whose signature is disjoint from the
signature of \Omc, and \Mmc is a set of annotated mappings, connecting
\dataSchema to \Omc \cite{XCKL*18}.
% \dl{Sorry, wasn't it annotated mappings? And wasn't them CQs??}. Ana:
% changed, thanks
The pair $(\obdaSpecification,\dataInstance)$ of an annotated OBDA
specification \obdaSpecification %$\obdaSpecification=(\Omc,\Mmc,\dataSchema)$
and an annotated data instance
\dataInstance is an \emph{annotated OBDA instance}.
In OBDA, data sources and mappings induce virtual
assertions. In annotated OBDA,
virtual assertions are annotated with the provenance
information of the mapping and of matching tuples in the data instance.
% inducing it.
Formally,
an \emph{annotated assertion} $(\assertionName(\vec{a}),p)$  is an expression
of the form $(A(a),p)$
or $(R(a,b),p)$, with $A\in \NC$, $R\in \NR$, $a,b\in \NI$, and $p\in\NM$.
 We write $ \varphi(\mu(\vec{x},\vec{y},\vec{z}))\subseteq \dataInstance$
 if $\mu$  is  a function  mapping  $\vec{x},\vec{y}$
 to \NI, $\vec{z}$ to \semiringVariables,  %variables? introduce a symbol for them?
 and $(\mu(\vec{t},t))\in P^\dataInstance$\negmedspace,
for every atom $P(\vec{t},t)$ in $\varphi(\vec{x},\vec{y},\vec{z})$.
Given an annotated mapping  \Mmc and %an annotated
data instance   \dataInstance,
the set $\mapping(\dataInstance)$ of  annotated assertions
\[
  \textstyle
  (\assertionName(\mu(\vec{x})), \ p\otimes   \prod_{z\in\vec{z}}\mu(z)),
  \text{ satisfying }
\]
$(\assertionName(\vec{x})\leftarrow\varphi(\vec{x},\vec{y},\vec{z}),\
p)\in\Mmc$ and $ \varphi(\mu(\vec{x},\vec{y},\vec{z}))\subseteq \dataInstance$
is the set of \emph{virtual annotated assertions} for \Mmc over \dataInstance.

% We denote by $\mapping(\dataInstance)$ the set of all such assertions.

The semantics of annotated OBDA instances is based on interpretations
over the signature of the ontology,
extending classical \DLLITE interpretations
to track provenance, when relevant.
An \emph{annotated interpretation} is a triple
$\Imc=(\Delta^\Imc,\Delta^\Imc_\sfm,\cdot^\Imc)$
where $\Delta^\Imc$ and $\Delta^\Imc_\sfm$ are non-empty disjoint sets (called the
\emph{domain} of \Imc and the
\emph{domain of monomials} of \Imc, respectively),
and $\cdot^\Imc$ is the \emph{annotated interpretation function} mapping
\begin{itemize}[noitemsep]
\item every $a\in\NI$ to some $a^\Imc\in\Delta^\Imc$;
\item  every $p,q\in \NM$
to some $p^\Imc,q^\Imc\in\Delta^\Imc_\sfm$
s.t.\ $p^\Imc=q^\Imc$ iff the
monomials $p$ and $q$ are mathematically equal
(modulo associativity and commutativity,
e.g., $(p\otimes q)^\Imc =  (q\otimes p)^\Imc$
by commutativity);
\item   every $A\in\NC$ to
some $A^\Imc\subseteq \Delta^\Imc\times \Delta^\Imc_\sfm$; and
\item every   $R\in\NR$ to
some $R^\Imc\subseteq \Delta^\Imc\times\Delta^\Imc\times \Delta^\Imc_\sfm$.
\end{itemize}
We extend
% the mapping
$\cdot^\Imc$ to further \DLLITE expressions as natural:
% in the natural way:
\[
  \begin{array}{rcl}
    (R^-)^\Imc &=& \{(e,d,p^\Imc)\mid (d,e,p^\Imc)\in R^\Imc\} \text{ ,}\\
    (\neg S)^\Imc &=&
    (\Delta^\Imc\times\Delta^\Imc\times \Delta^\Imc_\sfm)\setminus S^\Imc, \\
    (\exists S)^\Imc &=&
    \{(d,p^\Imc)\mid \exists e\in\Delta^\Imc:(d,e,p^\Imc)\in S^\Imc\} \text{ ,
     and}\\
    (\neg B)^\Imc &=& (\Delta^\Imc \times \Delta^\Imc_\sfm)\setminus B^\Imc.
  \end{array}
\]
The annotated interpretation \Imc \emph{satisfies}:
\[
%\\[0.5mm]\centerline{$
  \begin{array}{ll}
    \pair{A(a)}{p}, &\text{if } (a^\Imc,p^\Imc)\in A^\Imc;\\
    \pair{R(a,b)}{p}, &\text{if } (a^\Imc,b^\Imc,p^\Imc)\in R^\Imc;\\
    \pair{B\sqsubseteq C}{p}, &\text{if, for all } q\in\NM,
    (d,q^\Imc)\in B^\Imc\\
    & \qquad\text{implies that } (d,(q\otimes p)^\Imc)\in C^\Imc; \text{ and}\\
    \pair{S\sqsubseteq T}{p}, &\text{if, for all } q\in\NM,
    (d,e,q^\Imc)\in S^\Imc\\
    & \qquad\text{implies that } (d,e,(q\otimes p)^\Imc)\in T^\Imc.
  \end{array}
%$}\\[0.5mm]
\]
\Imc satisfies an annotated ontology $\Omc$, in symbols $\Imc\models\Omc$, if
it satisfies all annotated axioms in \Omc.  \Imc satisfies an annotated OBDA
instance $((\Omc,\mapping,\Smc), \dataInstance)$ if $\Imc\models\Omc$ and
$\Imc\models \mapping(\dataInstance)$.

\begin{example}
Consider the OBDA instance of Example~\ref{exa:mapping}
and an annotated interpretation \Imc with %domain
$\Delta^\Imc = \{ \ex{Renier}, \ex{Venice},\ex{Brugnaro}\}$,
$\Delta^\Imc_\sfm $ containing
% \davide{``the elements of the set'', or remove the set completely..}
${p\otimes n}$, ${q\otimes n}$, ${p\otimes m}$,
${q\otimes m}$, ${p\otimes n\otimes s}$, ${q\otimes n \otimes s}$, with
such individuals and monomials interpreted by themselves, and
\[
  \begin{array}{rcl}
    \ex{headGov}^\Imc &=&
    \{(\ex{R},\ex{V},p\otimes n),(\ex{B}, \ex{V},q\otimes n)\},\\
    \ex{Mayor}^\Imc &=&
    \{(\ex{R},p\otimes n \otimes s),(\ex{B},q\otimes n\otimes s)\},\\
    \ex{City}^\Imc &=& \{(\ex{V},p\otimes m),(\ex{V},q\otimes m)\}.
  \end{array}
\]
% \begin{itemize}[noitemsep]
% \item
%   $\ex{headGov}^\Imc=\{(\ex{R},\ex{V},p\otimes n),(\ex{B}, \ex{V},q\otimes
%   n)\}$,
% \item
%   $\ex{Mayor}^\Imc=\{(\ex{R},p\otimes n \otimes s),(\ex{B},q\otimes n\otimes
%   s)\}$,
% \item $\ex{City}^\Imc=\{(\ex{V},p\otimes m),(\ex{V},q\otimes m)\}$.
% \end{itemize}
\Imc is a model of such OBDA instance,
where  $\ex{R}$, $\ex{V}$, and $\ex{B}$ stand for $\ex{Renier}$, $\ex{Venice}$,
and $\ex{Brugnaro}$, respectively.
\end{example}
Following the database approach~\cite{Green07-provenance-seminal,GreenT17}, 
we annotate facts in interpretations with provenance information. 
However, in Green \textit{et al.}'s setting, the database ``is'' the (only) interpretation, 
while in our case we adopt the open world assumption (as in OBDA), so 
the semantics is based on multiple interpretations. 
%Hence, annotating facts in interpretations 
%seems a natural extension towards incomplete information (as in OBDA). 
%Indeed, 
Our semantics ensures that, if we have a tuple $(d,p^\Imc)\in C^\Imc$ and 
$(C \sqsubseteq D)$ is annotated with $n$, then $(d,(p\times n)^\Imc) \in D^\Imc$. So 
 derivations are also represented in interpretations, and thus can be entailed.
Each derivation is independent of the others.

Regarding the semantics of negation, 
we point out that, at the level of an interpretation, the lack of provenance information is a support for the negation of a fact. 
This apparent counterintuitive behaviour 
 %this 
 does not hold in all interpretations, 
 hence  %this behavior would 
 it does 
 not manifest in the entailments.
In fact, our focus in this paper is \emph{query} entailment (defined next), 
negations  are only defined to comply with the usual 
syntax and semantics of \DLLITE. 
They do not affect query results, as in \DLLITE.
 
%only positive inclusions are considered.
%negation we adopt the intuitive definition of negation. 

\mypar{Annotated Queries}
We extend the notion of conjunctive queries in DLs by allowing binary and ternary predicates,
%We define conjunctive queries  over binary and ternary predicates
where the last term of a tuple 
%is either a variable 
%in \NV \todo{GX:Do we need $\NV$, as $\NV \subset \NM$? Or we mean a 'variable' for monomials? 
%It seems to me that we are overloading the notion of 'variable'} 
%or
%an element of $X$, the set of variables of the polynomials
%a monomial in \NM
may contain provenance information represented as a monomial
(by definition of the semantics of annotated OBDA instances,
the last element of a tuple can only contain monomials, not sums).
More specifically,
a \emph{Boolean conjunctive query (BCQ)} $q$ is a sentence   %of the form
$\exists \vec{x}.\varphi(\vec{x},\vec{a},\vec{p})$,
where $\varphi$ is a conjunction of (non-repeating) atoms of the form $A(t_1,t)$,
$R(t_1,t_2,t)$, and $t_i$ is either an individual name from $\vec{a}$,
or a variable from $\vec{x}$, and $t$ (the last term of each tuple)
is either an element of $\NM$ in the list $\vec{p}$ %(built with elements of  $X$)
or a variable from $\vec{x}$.
We often write $P(\vec{t},t)$ to refer to an atom which can be either $A(t_1,t)$
or $R(t_1,t_2,t)$ and $P(\vec{t},t)\in q$ if $P(\vec{t},t)$ is an atom occurring in $q$.

A \emph{match} of the BCQ $q=\exists \vec{x}.\varphi(\vec{x},\vec{a},\vec{p})$ in the
annotated interpretation \Imc
is a function  $\map:\vec{x}\cup\vec{a}\cup \vec{p}\to\Delta^\Imc\cup\Delta^\Imc_{\sfm}$,
such that $\map(b)=b^\Imc$, for all $b\in \vec{a}\cup\vec{p}$,
and $\map(\vec{t},t)\in P^\Imc$, for every   $P(\vec{t},t)\in q$.
\Imc satisfies the BCQ $q$, written $\Imc\models q$,
if there is a match of $q$ in \Imc.
A BCQ is \emph{entailed by} an annotated OBDA instance
%$(\obdaSpecification,\dataInstance)$
if it is satisfied by every model of it. %$(\obdaSpecification,\dataInstance)$.
For a BCQ $q$ and an interpretation \Imc, $\nu_\Imc(q)$ denotes
the set of all matches   of $q$ in \Imc. % such that $\map(a)=a^\Imc$ for all $a\in \vec{a}$.
The \emph{provenance} of $q$ on \Imc, denoted $\p{\Imc}{q}$, is the (potentially infinite) expression: % the value:
\begin{align*}
 \textstyle \sum_{\pi\in\nu_\Imc(q)}\prod_{P(\vec{t},t)\in q} \pi^-(t) %\p{\Imc}{P(\map(\vec{t}))
\end{align*}
 where %$ \p{\Imc}{P(\map(\vec{t},t))}$
 $\pi(t)$
 is the last element of
 the tuple $\map(\vec{t},t)\in P^\Imc$; and
 $\pi^-(t)$ is any $v\in\NM$ s.t.\ $v^\Imc=\pi(t)$.
 For $p\in\NPr$, we write $p\subseteq \p{\Imc}{q}$ if
 $p$ is a sum of monomials and for each occurrence of a monomial in
 $p$ we find an occurrence of it in $\p{\Imc}{q}$.
\Imc \emph{satisfies} $q$ with provenance $p\in\NPr$,
written $\Imc\models (q,p)$, if
    $\Imc\models q$ and $p\subseteq\p{\Imc}{q}$.
The annotated OBDA instance
$(\obdaSpecification,\dataInstance)$ \emph{entails} $q$,
$(\obdaSpecification,\dataInstance)\models q$, if for all annotated interpretations
\Imc,  if $\Imc\models (\obdaSpecification,\dataInstance)$ then $\Imc\models q$; and
%$(\obdaSpecification,\dataInstance)$ \emph{entails} $(q,p)$, in symbols
$(\obdaSpecification,\dataInstance)\models (q,p)$,
if $(\obdaSpecification,\dataInstance)\models q$ and $p\subseteq\p{\Imc}{q}$,
for all %annotated interpretations
 \Imc satisfying $(\obdaSpecification,\dataInstance)$.

In our syntax, the atoms of the queries contain an additional parameter 
which may either be a variable or a monomial. 
As a result, %our query language is 
one can filter query results based on provenance information by specifying constraints 
%specifying a monomial 
in the last parameter of the atoms,
which was not possible in the original 
approach by Green \textit{et al.}~\cite{Green07-provenance-seminal,GreenT17}. 
 For example, $\exists xy. A(x,p)\wedge B(x,y)$ 
can be used to specify that we are only 
interested in matches of the query where the first atom is associated 
with a particular provenance. 
Variables can also be repeated, e.g.   
$\exists xy. A(x,y)\wedge B(x,y)$. 
One can fall back to the original  setting from databases, 
where no constraints are imposed, by simply associating 
the last term of each atom with a fresh variable (see  
  standard queries in Section~\ref{sec:idem}).

The \emph{size} $|X|$   of
  an annotated OBDA instance, a polynomial or a BCQ $X$
is the length of the string that represents $X$.
We assume a binary encoding of
elements of \NC,\NR,\NI and \NPr occurring in $X$.
We may omit `annotated' in front of terms such as `OBDA,'
`queries,' `inclusions,' `assertions,' and others,
 whenever  this is clear from the context.

\mypar{Reasoning Problems}
%It follows from our definitions that 
%satisfiability of an annotated OBDA instance 
Annotating OBDA instances with  provenance information   does not impact 
consistency checking. That is, an annotated OBDA instance is satisfiable 
precisely when the OBDA  instance that results from removing the annotations 
is satiafiable.
 %(i.e. simply remove the annotations and perform the check). 
We thus focus on  the problem of
\emph{query entailment} w.r.t.\ a provenance polynomial:
given an (annotated) OBDA instance $(\obdaSpecification,\dataInstance)$,
a query $q$ and a polynomial $p\in\NPr$
decide if $(\obdaSpecification,\dataInstance)\models (q,p)$.
Another important and related problem is to compute the provenance of a query:
%We define the problem of computing
%the provenance of a query  as   follows:
 given an   OBDA instance $(\obdaSpecification,\dataInstance)$ %in an ontology language \Lmc
and a query $q$, compute
the set of all $p\in\NPr$ such that $(\obdaSpecification,\dataInstance)\models (q,p)$.
In our formalism, the latter problem  depends on whether
there is a finite  set of polynomials which we can compute.
%\ano{changed to say `set` because the polynomials themselves are finite (those entailed, otherwise they would not
%be in \NPr or we would need complete semirings instead of just semirings, to have infinitary elements) but the set is infinite  }
As shown next, in  \DLLITE
the set of provenance polynomials may be infinite.
\begin{example}
\label{ex:multiplication}
Consider an OBDA instance
% $((\Omc,\Mmc,\dataSchema),\dataInstance)$ as in Ex.~\ref{exa:mapping}, and
% let $\Omc'=\Omc\cup\{\pair{\ex{Mayor}\sqsubseteq\exists\ex{headGov}}{t}\}$.
$(\Pmc,\dataInstance)$ as in Ex.~\ref{exa:mapping}, but where now $\Omc$ of
$\Pmc$
contains also $\pair{\ex{Mayor}\sqsubseteq\exists\ex{headGov}}{t}$.  For all
$i\in\mathbb{N}$, % it holds that
$(\Pmc,\dataInstance)\models \pair{\ex{Mayor}(\ex{Renier})}{p\otimes n\otimes
 s^{i+1}\otimes t^i}$.  Indeed, for any model \Imc of $(\Pmc,\dataInstance)$,
% satisfies the axioms, so
$(\ex{Renier},(p\otimes n\otimes s)^\Imc)\in \ex{Mayor}^\Imc$ implies
$(a,(p\otimes  n\otimes s\otimes t)^\Imc)\in (\exists\ex{headGov})^\Imc$,
which implies
$(\ex{Renier},(p\otimes n\otimes s^2\otimes t)^\Imc)\in\ex{Mayor}^\Imc$,
and so on.
\end{example}

In Section~\ref{sec:rewriting}
we consider  the problem of
query entailment w.r.t.\ a provenance polynomial.
Note that in Example~\ref{ex:multiplication},
if the semiring is multiplicatively idempotent (i.e., $s\otimes s=s$), 
the set of provenance polynomials is finite: the only polynomial is $p\otimes n\otimes s\otimes t$.
This is not a coincidence; under multiplicative-idempotency, the set of provenance polynomials is always
finite. The following proposition states that multiplicative-idempotency
is indeed 
sufficient to guarantee a  finite set of polynomials.
\begin{restatable}{proposition}{PropositionFinite}\label{prop:finite}
Under multiplicative idempotency, for any satisfiable   OBDA instance
$(\obdaSpecification,\dataInstance)$ and
 BCQ $q$, the set of polynomials $p\in\NPr$ such that $(\obdaSpecification,\dataInstance)\models (q,p)$ is finite.
 %Moreover, \DLLITEp enjoys the finite model property.
\end{restatable}

In Section~\ref{sec:idem} we study  idempotent semirings and
 consider the problem of computing
the provenance of a query.
%\ano{explain that set of polynomials of a queqy in this case
%can be treated as  the polynomial of a query?}\rpn{I would leave it as it is now} Ok!

%%% Local Variables:
%%% mode: latex
%%% TeX-master: "paper"
%%% End:

% !TEX root =  paper.tex
\section{Provenance Annotated Query Entailment}\label{sec:rewriting}  

%In this section 
We establish complexity results for the problem 
of deciding whether an   OBDA  instance  entails 
a (provenance annotated) query. % $(q,p)$. 
For clarity of presentation, we split our proof  in two parts.
We first show that  
for an   OBDA  instance $(\obdaSpecification,\database)$
and a query $(q,p)$, 
there is an  OBDA  instance $(\obdaSpecification_m,\database_m)$
and a set ${\sf Tr}(q_m,p_m)$ of
(non-annotated) queries  
such that $(\obdaSpecification,\database)\models (q,p)$
iff
$(\obdaSpecification_m,\database_m)$    entails some 
   $q'\in{\sf Tr}(q_m,p_m)$. 
   Moreover, the sizes of $(\obdaSpecification_m,\database_m)$ and 
   $q'$ are polynomial in the sizes of 
   $(\obdaSpecification,\database)$
and $(q,p)$. 
Then, we adapt the query rewriting 
  algorithm \perfectRef~\cite{dl-lite} to decide whether 
  $(\obdaSpecification_m,\database_m)\models q'$.

\mypar{Part 1: Characterization}
%
%~ Given a function $\dagger:\NV\rightarrow\NV$ %, called \NV-function,
%~ and an OBDA instance $(\obdaSpecification_m,
%~ \database_m)$, we write $(\obdaSpecification^\dagger_m,
%~ \database^\dagger_m)$ for the result of  simultaneously replacing
%~ each variable $p\in \NV$ occurring in $(\obdaSpecification_m,
%~ \database_m)$
%~ by $\dagger(p)$, and similarly for annotated queries 
%~ $(q_m,p_m)$ (where we also replace 
%~ occurrences of variables in monomials).
%
Lemma~\ref{lem:not-equal-marked} states that, 
given %the connection between 
an  OBDA instance $(\obdaSpecification,\database)$ and 
a query $(q,p)$,
there is an  OBDA instance $(\obdaSpecification_m,\database_m)$ 
and  a query $(q_m,p_m)$ that can
be used to decide %whether
%such that %entailment of $(q,p)$ w.r.t. 
$(\obdaSpecification,\database)\models (q,p)$ 
%iff 
%can be decided 
%using  
 %$(\obdaSpecification_m,\database_m)\models(q_m,p_m)$ 
 and, moreover,  
%satisfying the property that if it entails 
%an annotated query then 
all monomials in $p_m$ are mathematically 
distinct (modulo associativity, commutativity,  and distributivity). 
%~ Given a function $\dagger:\NV\rightarrow\NV$, called \NV-function,
%~ and an OBDA instance $(\obdaSpecification_m,
%~ \database_m)$, we write $(\obdaSpecification^\dagger_m,
%~ \database^\dagger_m)$ for the result of  simultaneously replacing
%~ each variable $p\in \NV$ occurring in $(\obdaSpecification_m,
%~ \database_m)$
%~ by $\dagger(p)$, and similarly for annotated queries $(q_m,p_m)$.
\begin{restatable}{lemma}{Lemmanotequalmarked}\label{lem:not-equal-marked}
Given a satisfiable  OBDA instance 
$(\obdaSpecification,\dataInstance)$ and a query $(q,p)$, 
there are 
 $(\obdaSpecification_m,\dataInstance_m)$ and 
   $(q_m,p_m)$
    such that
\begin{itemize}[noitemsep]
%\item $(\obdaSpecification,\database)=(\obdaSpecification^\dagger_m,
%\database^\dagger_m)$, $(q,p)=(q^\dagger_m,p^\dagger_m)$;
\item 
  any two monomials $p_1$, $p_2$ appearing in $p_m$
  % , it holds that $p_1$ and $p_2$
  are mathematically distinct;
\item $(\obdaSpecification,\dataInstance)\models (q,p)$ iff 
$(\obdaSpecification_m,\dataInstance_m)\models (q_m,p_m)$;  and
\item $|(\obdaSpecification_m, \database_m)|+|(q_m,p_m)|$ is polynomially bounded by  
$|(\obdaSpecification,\database)|+|(q,p)|$. %, respectively.
\end{itemize}
\end{restatable}
We show that, given  $(\obdaSpecification_m,\dataInstance_m)$ 
and  $(q_m,p_m)$ as in
Lemma~\ref{lem:not-equal-marked},
% a BCQ $q$ and 
%a polynomial $p\in\NPr$,
 $(q_m,p_m)$ can be translated
%the query $q$ and the polynomial 
%$p$ 
into a set of %\emph{non-annotated} 
queries such that $(\obdaSpecification_m,\dataInstance_m)$ 
entails $(q_m,p_m)$ iff it 
entails at least one of these queries. %from this set.  
We first define the translation of a BCQ where all terms are variables (no individual names 
and no polynomials), and then adapt the translation for
 the general case. 
Given the BCQ $q_m=\exists \vec{x}.\ \varphi(\vec{x})$ with $k$ atoms and 
   $p_m\in \NPr$ with $n$ monomials,
define $\tr{q_m}{p_m}$ as the set of all BCQs: % of the form:  
\begin{equation}\label{eq:trans}
\textstyle \exists \vec{y}.\    \bigwedge_{1\leq i\leq n}\varphi_i(\vec{x_i}),
\end{equation} 
where $\vec{y}=\vec{x_1},\ldots, \vec{x_n}$ and 
each $q_i=\exists \vec{x_i}.\ \varphi_i(\vec{x_i})$ is a `copy' of $q$ 
in which we replace each variable $x\in\vec{x}$ by a fresh variable $x_i\in\vec{x_i}$. 
We check whether we can find 
  the monomials of the polynomial in these matches %Assume $q$ has $m$ atoms. 
by replacing the last variable in each $j$-th atom 
  of $q_i$ by a monomial $p_{i,j}\in \NM$ built from symbols 
  occurring in $p_m$ such that $\prod_{1\leq j\leq k}  p_{i,j}= p_i$
  for some $p_i\in\NPr$, 
with $1\leq i\leq n$; and  $\sum_{1\leq i\leq n}  p_i= p$.
 
 The translation of a BCQ with individual names is similar, 
 except that we must add such individual names 
 in each copy of the query; that is,     
 we would replace the corresponding variable in the translation with 
 the individual name occurring in the query. 
Theorem~\ref{thm:trans} formalises the correctness of our translation, where 
we write $(\obdaSpecification,\dataInstance)\models \tr{q}{p}$ 
%(resp. $\Imc\models \tr{q_m}{p_m}$) 
to express that
there is $q'\in \tr{q}{p}$ such 
that $(\obdaSpecification,\dataInstance)\models q'$. % (resp. $\Imc\models q'$).  
%~ In the following, we write $(\obdaSpecification_m,\dataInstance_m)\models \tr{q_m}{p_m}$ 
%~ %(resp. $\Imc\models \tr{q_m}{p_m}$) 
%~ to express that
%~ there is $q'\in \tr{q_m}{p_m}$ such 
%~ that $(\obdaSpecification_m,\dataInstance_m)\models q'$ (resp. $\Imc\models q'$).  
%
%~ \begin{example}
%~ Consider the query $q=\exists xyzw. (\ex{headGov}(x,y,z)\wedge \ex{City}(y,w))$
  %~ and the polynomial $p=(s\otimes t)\oplus (s\otimes r)$. 
%~ Then, 
%~ $$\exists x_1 y_1 x_2 y_2. (\ex{headGov}(x_1,y_1,s)\wedge \ex{City}(y_1,t)\wedge 
%~ $$$$
%~ \ex{headGov}(x_2,y_2,s)\wedge \ex{City}(y_2,r))$$
%~ is in $\tr{q}{p}$.
%~ %$\exists x_1 x_2 . A(x_1,s)\wedge B(x_1,t)\wedge A(x_2,s)\wedge B(x_2,r)\in\tr{q}{p}$.
%~ \end{example}
%\ano{changing}
\begin{example}
Consider the query $$q=\exists xyzw. ({\sf headGov}(x,y,z)\wedge {\sf City}(y,w))$$
  and the polynomial $p=(s\otimes t)\oplus (s\otimes r)$. 
Then, 
\[
  \exists x_1 y_1 x_2 y_2. (
  \begin{array}[t]{@{}l}
    {\sf headGov}(x_1,y_1,s)\wedge {\sf City}(y_1,t) \wedge{}\\
    {\sf headGov}(x_2,y_2,s)\wedge {\sf City}(y_2,r))
  \end{array}
\]
is in $\tr{q}{p}$.
%$\exists x_1 x_2 . A(x_1,s)\wedge B(x_1,t)\wedge A(x_2,s)\wedge B(x_2,r)\in\tr{q}{p}$.
\end{example}
\begin{restatable}{theorem}{Theoremtranslation}\label{thm:trans} 
Let $(\obdaSpecification,\dataInstance)$ be an  OBDA instance, 
$q$ a BCQ and $p \in \NPr$ a polynomial %such that 
formed of mathematically distinct monomials. 
%Then, 
%$\Omc\models (q,p)$ if, and only if, 
%$\Omc\models \tr{q}{p}$. %\rpn{I removed the claim from the proof, I hope you don't mind :)}  it is better thanks
$(\obdaSpecification,\dataInstance)\models (q,p) \text{ iff }
(\obdaSpecification,\dataInstance)\models \tr{q}{p}.$ 
\end{restatable} 
Without assuming that $p \in \NPr$ is %a polynomial %such that 
formed of mathematically distinct monomials, we would need 
to add inequalities to the queries in $\tr{q}{p}$
(there is  no way to distinguish $\tr{q}{p+p}$
from $\tr{q}{p}$). 
%does not need 
%inequalities to distinguish entailment of $(q,p+p)$ from $(q,p)$
%translation 
%. 
By Lemma~\ref{lem:not-equal-marked}, given 
the   OBDA instance $(\obdaSpecification,\dataInstance)$ 
and query $(q,p)$, 
there are $(\obdaSpecification_m,\dataInstance_m)$ 
and  $(q_m,p_m)$, satisfying the assumption of Theorem~\ref{thm:trans}, 
which we can use to decide whether 
$(\obdaSpecification,\dataInstance)\models (q,p)$. 
This is crucial for query entailment
since entailment of conjunctive queries with inequalities 
in \DLLITE is undecidable~\cite{Guti}.  

\mypar{Part 2: Query Rewriting}
%Finally, 
We adapt the classical query rewriting algorithm 
\perfectRef~\cite{dl-lite} to decide whether 
$(\obdaSpecification,\dataInstance)\models q'$,
for $q'\in \tr{q}{p}$, where $(\obdaSpecification,\dataInstance)$
and $(q,p)$ are as in Theorem~\ref{thm:trans}.
%deal with provenance 
%annotated OBDA instances and queries. % in \DLLITEp. 
%Given a BCQ $q$, 
% 
When possible, we use the definitions and terminology from \cite[Sec. 5.1]{dl-lite}, 
adapting some of them to our setting if needed. 

For simplicity, for each role $R^-$ occurring in an OBDA instance
$((\Omc,\Mmc,\dataSchema),\dataInstance)$, we add to \Omc 
the annotated role inclusions $(R^-\sqsubseteq \overline{R},p_R)$
 and $(\overline{R}\sqsubseteq R^-,p'_R)$, where $\overline{R}$ is a fresh role name and 
 $p_R,p'_R$ are fresh variables of a provenance semiring.  
We   assume w.l.o.g.\ that inverse roles only occur in such role inclusions 
by replacing other occurrences of $R^-$ with $\overline{R}$.
% In the following,
The symbol $``_{-}"$ denotes non-distinguished non-shared variables. 
A positive inclusion $I$ is a provenance annotated role or concept inclusion without negations. 
$I$ is \emph{applicable} to $A(x,p)$ if $I$ is annotated 
with $v$ occurring in $p$ and it has $A$ in its right-hand side.
A positive inclusion $I$ is applicable to $R(x, y,p)$ if
\begin{enumerate*}[label=\textit{(\roman*)}]
\item $x=_-$, $I$ is annotated with $v$ occurring in $p$, and the
  right-hand side of $I$ is $\exists R$, or
\item $I$ is a role inclusion annotated with $v$ occurring in $p$ and its
  right-hand side is $R$ or~$R^-$.
\end{enumerate*}
Given $p\in\NM$ and   $v\in  \NV$   occurring in $p$, we denote by
$\pminus{p}{v}$ the result of removing one occurrence of $v$ from $p$.

\begin{definition}\label{def:atom}
Let $g$ be an atom and $I$ a positive inclusion applicable to $g$.
The atom obtained from $g$ by applying $I$, denoted by $gr(g,I)$, is 
defined as follows: 
\begin{itemize}[noitemsep] 
\item $gr(A(x,p ),(A_1\sqsubseteq A,v))=A_1(x,\pminus{p}{v})$; %, if $g=A(x,p )$ and $I=(A_1\sqsubseteq A,v)$;
\item $gr(A(x,p ),(\exists R\sqsubseteq A,v))=R(x,_-,\pminus{p}{v})$; %, if $g=$ and $I=$;
\item $gr(R(x,_-,p ),(A\sqsubseteq \exists R,v))=A(x,\pminus{p}{v})$; %, if $g=)$ and $I=$;
\item $gr(R(x,_-,p ),(\exists R_1\sqsubseteq \exists R,v))=R_1(x,_-,\pminus{p}{v})$; %, if $g=$ and $I=(\exists R_1\sqsubseteq \exists R,v)$;
\item $gr(R(x,y,p ),(R_1\sqsubseteq R,v))=R_1(x,y,\pminus{p}{v})$; %, if $g=$ and $I= $;  
\item $gr(g,I)=R_1(y,x,\pminus{p}{v})$, if $g=R(x,y,p )$ and either 
		$I= (R_1\sqsubseteq R^-,v)$ or $I= (R_1^-\sqsubseteq  R,v)$. 
\end{itemize}
\end{definition}

We  use \perfectRef (Algorithm~\ref{alg:perfectRef}) originally presented in \cite{dl-lite}, 
except that
the applicability of a positive inclusion $I$ to an atom $g$ is as previously described
and $gr(g,I)$ follows Definition~\ref{def:atom}. 
\begin{algorithm}[tb]
\caption{\perfectRef}\label{alg:perfectRef}
\textbf{Input:} a BCQ $q$, a set of positive inclusions $\Omc_\Tmc$\\
\textbf{Output:} a set of BCQs $\PR$
\begin{algorithmic}[1]
\STATE $\PR:=\{q\}$
\REPEAT
	\STATE $\PR':= \PR$
	\FORALL{$q\in \PR'$, \textbf{all} $g,g_1,g_2\in q$ and \textbf{all} $I\in \Omc_\Tmc$}
%		\FORALL{$g\in q$ }
%			\FORALL{$I\in \Omc_\Tmc$}
			\IF{$\{q[g/gr(g,I)]\}\notin \PR$ and $I\in \Omc_\Tmc$ is applicable to $g\in q$}
			\STATE $\PR:= \PR\cup \{q[g/gr(g,I)]\}$ \label{ln:add}
			\ENDIF
%			\ENDFOR
%		\ENDFOR
	%	\FORALL{$g_1,g_2\in q$}
		\IF{there are $g_1,g_2\in q$ such that $g_1$ and $g_2$ unify}
			\STATE $\PR:= \PR\cup \{\tau(\mn{reduce}(q, g_1, g_2))\}$
		\ENDIF
	%\ENDFOR
	\ENDFOR
\UNTIL{{$\PR'= \PR$}}\\
\RETURN $\PR$
\end{algorithmic}
\end{algorithm}
Let $q[g/g']$ denote  the BCQ obtained from $q$ by
 replacing the atom $g$ with a new atom $g'$;
 let~$\tau$ be a function that takes as input a BCQ $q$ and returns
a new BCQ obtained by replacing each occurrence of an unbound 
variable in $q$ with the symbol `$_-$'; and
let ${\sf reduce}$ be a function that takes as input a BCQ $q$ and
two atoms $g_1$, $g_2$ and returns the result of applying to $q$ the most
general unifier of $g_1$ and~$g_2$ (unifying mathematically equal terms).
$\perfectRef(q,\Omc_\Tmc)$ is the output of the algorithm
\perfectRef  over $q$ (with a monomial in $\NM$ in the last 
parameter of each atom) and a set $\Omc_\Tmc$ of positive inclusions of 
an OBDA instance $((\Omc,\Mmc,\dataSchema),\dataInstance)$. 
%as input. 
\begin{example}
Consider an OBDA instance $((\Omc,\Mmc,\dataSchema),\dataInstance)$
as in Ex.~\ref{exa:mapping}.
We call Algorithm~\ref{alg:perfectRef}
with $\Omc_\Tmc$ and the query $q=\exists x. \ex{Mayor}(x,p\otimes n\otimes s)$  as input. 
 Since  $I$ is applicable to 
$g=\ex{Mayor}(x,p\otimes n\otimes s)$, in Line~\ref{ln:add}, 
Alg.~\ref{alg:perfectRef} adds to $\PR$  the result of replacing
$g$ by $gr(g,I)=\ex{headGov}(x,_-,p\otimes n)$ in $q$.
Hence, 
$q^\ddagger=\exists x,y.\ \ex{headGov}(x,y,p\otimes n)\in\perfectRef(q,\Omc_\Tmc).$  
Indeed $q^\ddagger$ is a rewriting of $q$.
\end{example}

Our next theorem states the correctness of  Algorithm~\ref{alg:perfectRef}. 
\begin{restatable}{theorem}{Theoremqueryrewriting1}\label{thm:combined1}
  Let $q$ be a BCQ and $\Omc_\Tmc$ the set of positive inclusions of an OBDA 
  specification $\Pmc=(\Omc,\Mmc,\Smc)$.
  Given $q$ and $\Omc_\Tmc$ as input, 
  Algorithm~\ref{alg:perfectRef} terminates 
   and 
  outputs a set of BCQs $PR$ 
  such that, for all data instances \Dmc where   $(\Pmc,\dataInstance)$ 
  is satisfiable,
  $(\Pmc,\dataInstance)\models q$ iff there is $q^\ddagger\in PR$ 
  such that 
   $((\emptyset,\Mmc,\Smc),\dataInstance)\models q^\ddagger$.
\end{restatable}

Termination of our modified version of \perfectRef is analogous
to~\cite[Lemma 34]{dl-lite}, except that now the number of 
terms is exponential in the size of monomials occurring in the query, and thus 
in the size of the query. This is due to Definition~\ref{def:atom}, 
where we `break' the monomial into a smaller one. 
Our modification does not change the upper bounds obtained with the algorithm, 
since for data complexity the query is not part of the input and 
the upper bound for combined complexity, which we 
establish in Theorem~\ref{thm:combined}, is obtained by 
a non-deterministic version of the algorithm.

\begin{restatable}{theorem}{Theoremqueryrewriting}\label{thm:combined}
Answering provenance annotated queries w.r.t.\ OBDA instances is
\mbox{\NP-}com\-ple\-te (combined complexity).
\end{restatable}

%%% Local Variables:
%%% mode: latex
%%% TeX-master: "paper"
%%% End:

% !TEX root =  paper.tex
\section{Computing the Provenance of a Query }\label{sec:idem}

We now consider the problem of computing the provenance
of a query. To avoid the case of an infinite provenance, we focus on the special case
where the provenance semiring is fully idempotent, which is a sufficient condition for
finite provenance (Proposition~\ref{prop:finite}).
The semiring is \emph{fully idempotent} if
%That is, the operations $\otimes$ and $\oplus$
%are such that
for every polynomial $p\in\NPr$, $p\otimes p= p$ and $p\oplus p=p$.
This is the case, e.g., if the provenance refers
to the name of the source of the knowledge; having several times the same name does not affect the result.
Alternatively, one can model
% consider
access rights and observe whether
certain pieces of knowledge are needed for the entailment of a query w.r.t.\
an OBDA instance.

For fully idempotent semirings, the task corresponds to computing
relevant monomials. More precisely, in this special case we want to compute all monomials
$p$ such that $(\obdaSpecification,\dataInstance)\models (q,p)$.
The \emph{provenance of the query w.r.t.\ the OBDA instance} is the addition of all these monomials.
This definition is equivalent to the general one since the semiring is idempotent: repetitions of a monomial do
not affect the result, and repetitions of a variable within a monomial can be removed.
If the semiring is only multiplicatively idempotent, then computing monomials does not
suffice, as some of them may appear several times. However, the problem
is still simplified to find the (finite) number of repeated monomials to be observed.
In general, the query polynomial may be composed of exponentially many monomials, even if the
query is a simple one of the form $\exists x.A(a,x)$, with $A\in\NC$.
\begin{restatable}{proposition}{Propositionexponential}
\label{prop:exp}
There exists an %\ano{add annotated everywhere or say that we may omit `annotated'}
OBDA instance $(\Pmc,\dataInstance)$ and a simple query $q$ such that the provenance polynomial of
$q$ w.r.t.\ $(\Pmc,\dataInstance)$ is formed of exponentially many monomials.
\end{restatable}
For some queries, provenance cannot be expressed by a
provenance polynomial of polynomial length in the size of the ontology, even if
an expanded form is not required. This follows from known results
in monotone complexity~\cite{KaWi90}: there is
no monotone Boolean formula (i.e., propositional formula using only the connectives $\wedge$
and $\vee$)
of polynomial length expressing all the simple paths between two nodes in a graph.
This holds already for \emph{complete} graphs.
Graphs can be described in \DLLITE (and simpler logics) using basic inclusion
axioms, and monotone Boolean formulas are provenance polynomials over an idempotent
semiring, where the $\land$ and $\lor$ serve as product and addition.
Hence we have the following result~\cite{penaloza-09}.

\begin{restatable}{proposition}{PropIdempotent}\label{prop:idempotent}
There exist an OBDA instance  $(\Pmc,\dataInstance)$ and a query $q$ such that
the provenance   of $q$ w.r.t.\ $(\Pmc,\dataInstance)$
cannot be represented in polynomial space. This holds even for idempotent
semirings, and if every axiom has a unique label.
\end{restatable}
On the other hand, if every axiom is labeled with a unique variable,
then the provenance polynomial for \emph{instance queries} can be computed
efficiently, whenever its length does not increase greatly; that is, it can be computed in polynomial time in
the size of the input \emph{and the output}. The proof of this claim follows the same ideas
from~\cite{PeSe-AIJ17}, based on the fact that all the relevant monomials from the provenance are
enumerable with polynomial delay.

\begin{lemma}
\label{lem:poldelay}
The provenance $p$ of an instance query w.r.t.\ an OBDA instance $(\Pmc,\dataInstance)$
can be computed in polynomial time in the
size of $(\Pmc,\dataInstance)$  \emph{and of the polynomial $p$}.
\end{lemma}

%~ \todo{this paragraph could be removed if needed}
%~ A last problem to consider is \emph{relevance}: whether a given semiring variable appears in the polynomial
%~ of an answer. To solve this problem, we can modify the rewriting algorithm in a way that, instead of keeping
%~ track of all different ways in which an atom is derived, we preserve only their union, as a unique monomial.
%~ Hence, relevance can be solved within the same complexity bounds as provenance query answering.

\begin{algorithm}[tb]
\caption{\computeProv}\label{alg:computeProv}
\textbf{Input:} a   BCQ $q_0$, an OBDA instance $((\Omc,\Mmc,\Smc),\database)$\\
\textbf{Output:} the provenance $p$ of $q$ w.r.t. $((\Omc,\Mmc,\Smc),\database)$

\begin{algorithmic}[1]
%\STATE For all $P(\vec{t},t)\in q$, replace $t$ by $\star$
  \STATE   $PR:=\perfectRefIdem(q^{\star}_0,\Omc_\Tmc)$,
  %with $gr(g,I)$ as in Def.~\ref{def:atom2}
  \label{ln:notation}
  %where $\Omc_\Tmc$
  %are the positive inclusions in $\Omc$ and $gr(g,I)$
  %is as in Def.~\ref{def:atom2}
  \FORALL{$q\in PR$}
  \FORALL{ matches $\pi$ of $q_{\vec{y}}$ in $\Imc_{\Mmc(\dataInstance)}$} \label{ln:standard}
  \STATE $PR:= PR\cup\{q^{-\star}_{\vec{y},\pi}\}$ \label{ln:addmatch}
  \ENDFOR
  \STATE $PR:= PR\setminus\{q\}$ \label{ln:remove}
  \ENDFOR
 %\STATE Remove from $PR$ all $q$ not \consistent with $q_0$ \label{ln:consistent}
  \RETURN $p:=\sum_{q\in PR}\prod_{P(\vec{t},t)\in q}  t$ \label{ln:final}
\end{algorithmic}
\end{algorithm}

We give an algorithm  for computing the provenance of a BCQ
w.r.t.\ an OBDA instance.
We focus on BCQs that do not have monomials %from \NM
in the last term of the atom. %More specifically,
A BCQ  $q=\exists \vec{x}.\varphi(\vec{x},\vec{a})$ %\ano{small changes}
is \emph{standard} if, for all $P(\vec{t},t)\in q$,
$t$ is a fresh variable in $\vec{x}$.
%We show a procedure
Algorithm~\ref{alg:computeProv} computes the provenance of a standard BCQ
w.r.t.\ an OBDA instance.
%The procedure if describes a procedure for computing the %\ano{changing, section will be shortened}
%provenance of standard BCQ w.r.t. an OBDA instance.
We adopt the same notation used for describing \perfectRef~\cite{dl-lite}
(also used in Section~\ref{sec:rewriting}).
%However, %the notion of applicability of an inclusion w.r.t. an atom
%here
  \perfectRefIdem is a variant of \perfectRef
(Algorithm~\ref{alg:perfectRef}), where the notions of
applicability of an  inclusion $I$ w.r.t.\ an atom $g$ and the
definition of $gr(g,I)$ are as follows.
$I$ is applicable to an atom $A(x,p)$ if $I$
has $A$ in its right-hand side.
A positive inclusion $I$ is applicable to an atom $R(x, y,p)$ if
\begin{enumerate*}[label=\textit{(\roman*)}]
\item $x=_-$, and the right-hand side of $I$ is $\exists R$, or
\item the right-hand side of $I$ is either $R$ or $R^-$.
\end{enumerate*}
Given $p\in\NM$ and   $v\in  \NV$, we define
$\pplusidem{p}{v}$ as $\pplus{p}{v}$ if $v$ does not occur in $p$, and we define
$\pplusidem{p}{v}$ as $p$, otherwise. E.g., $\pplusidem{vw}{v}=vw$.

\begin{definition}\label{def:atom2}
Let $g$ be an atom and $I$ a positive inclusion applicable to $g$.
The atom obtained from $g$ by applying $I$, denoted by $gr(g,I)$, is
defined as follows:
\begin{itemize}[noitemsep]
\item $gr(A(x,p ),(A_1\sqsubseteq A,v))=A_1(x,\pplusidem{p}{v})$; %, if $g=A(x,p )$ and $I=(A_1\sqsubseteq A,v)$;
\item $gr(A(x,p ),(\exists R\sqsubseteq A,v))=R(x,_-,\pplusidem{p}{v})$; %, if $g=$ and $I=$;
\item $gr(R(x,_-,p ),(A\sqsubseteq \exists R,v))=A(x,\pplusidem{p}{v})$; %, if $g=)$ and $I=$;
\item $gr(R(x,_-,p ),(\exists R_1\sqsubseteq \exists R,v))=R_1(x,_-,\pplusidem{p}{v})$; %, if $g=$ and $I=(\exists R_1\sqsubseteq \exists R,v)$;
\item $gr(R(x,y,p ),(R_1\sqsubseteq R,v))=R_1(x,y,\pplusidem{p}{v})$; %, if $g=$ and $I= $;
\item $gr(g,I)=R_1(y,x,\pplusidem{p}{v})$, if $g=R(x,y,p )$ and either
                $I= (R_1\sqsubseteq R^-,v)$ or $I= (R_1^-\sqsubseteq  R,v)$.
\end{itemize}
\end{definition}

For standard BCQs,  Algorithm~\ref{alg:computeProv}
is sound and complete. Termination of Algorithm~\ref{alg:computeProv} is
an easy consequence of termination of  \perfectRef.
The main difference between Algorithm~\ref{alg:computeProv}
and Algorithm~\ref{alg:perfectRef} (Section~\ref{sec:rewriting}) is that
 here we assume that a standard BCQ is given (without any provenance information) and
 we aim at computing its provenance. Instead of removing variables of the semiring
 while applying positive inclusions (Definition~\ref{def:atom}),
   we add the variables of the semiring whenever
 the associated positive inclusion is applied (Definition~\ref{def:atom2}).
In Line~\ref{ln:notation}, we write $q^\star$ to denote the result
of replacing each $t$ in  $P(\vec{t},t)\in q$ by $\star$,
where $\star$ is a fresh symbol from \NV. This
transformation ensures that in Definition~\ref{def:atom2}
the last term is always an element of \NM.
In Line~\ref{ln:standard}, we denote by $q_{\vec{y}}$ the result of replacing,
for each $P(\vec{t},t)\in q$,
the last term $t$ by a fresh variable from $\vec{y}$ (i.e., $q_{\vec{y}}$
is a standard BCQ).
We perform another transformation in Line~\ref{ln:addmatch},
denoted by $q^{-\star}_{\vec{y,}\pi}$, which is
%We define $q^{-\star}_{\vec{y},\pi}$ as
the result of replacing, for each $P(\vec{t},t)\in q$,
the symbol $\star$ in $t$
  by $u\in\NM$ such that $u^\Imc=\pi(y)$ (if there are
  multiple mathematically equal such $u$, we simply choose $u$ arbitrarily), where $y$ is the
  last term of the corresponding atom in $q_{\vec{y}}$ (that is, $P(\vec{t},y)\in q_{\vec{y}}$).
  Observe that $\pi$ is a match of $q_{\vec{y}}$ in $\Imc_{\Mmc(\dataInstance)}$.

%In the following we provide an example  run of Algorithm~\ref{alg:computeProv}.
\begin{example}
Assume Algorithm~\ref{alg:computeProv} receives as input the standard query
$q_0 = \exists xz. {\sf Mayor}(x,z)$ and an %an annotated
OBDA instance $((\Omc,\Mmc,\Smc),\Dmc)$ with
$\Omc=\{(\exists {\sf headGov}\sqsubseteq {\sf Mayor},s)\}$ and
\[
  \Mmc(\Dmc) = \{
  \begin{array}[t]{@{}l}
    ({\sf headGov}({\sf Renier},{\sf Venice}),u),\\
    ({\sf headGov}({\sf Brugnaro},{\sf Venice}),v)\}.
  \end{array}
\]
In Line~\ref{ln:notation}, Algorithm~\ref{alg:computeProv}  calls
$\perfectRefIdem$, defined as  a variant of \perfectRef
(Algorithm~\ref{alg:perfectRef}), where the notions of
applicability of an  inclusion $I$ w.r.t. an atom $g$ and the
definition of $gr(g,I)$ are as in Section~\ref{sec:idem}. The return of
$\perfectRefIdem$ is $PR=\{\exists x. {\sf Mayor}(x,\star), \exists xz. {\sf headGov}(x,z,\star\otimes s)\}$.
 Then, for all $q\in PR$ and all matches $\pi$ of $q_{\vec{y}}$ in
 $\Imc_{\Mmc(\Dmc)}$ (if they exist) the algorithm adds $q^{-\star}_{\vec{y},\pi}$
 to $PR$. In this example, assume $q=\exists xz. {\sf headGov}(x,z,\star\otimes s)$.
 We have two matches %,  $\pi$ and $\pi'$,
 of  $q_{\vec{y}}=\exists xzy. {\sf headGov}(x,z,y)$ in $\Mmc(\Dmc)$,
 one mapping $y$ to $u$  (call this match $\pi$)
 and the other mapping $y$ to $v$ (call it $\pi'$).
 So, $q^{-\star}_{\vec{y},\pi}=\exists xz. {\sf headGov}(x,z,u\otimes s)$
 and $q^{-\star}_{\vec{y},\pi'}=\exists xz. {\sf headGov}(x,z,v\otimes s)$.
 In Line~\ref{ln:remove}, Algorithm~\ref{alg:computeProv} removes
 $q^\star_0$ from $PR$. Finally, in Line~\ref{ln:final},
 it returns the polynomial $u\otimes s \oplus v\otimes s$.
%\ano{working here}
\end{example}
\begin{restatable}{theorem}{TheoremComputingProv} %\ano{todo}
  Let $q$ be a standard BCQ and $(\Pmc,\dataInstance)$ an OBDA instance.
  Given $q$ and $(\Pmc,\dataInstance)$ as input to
  Algorithm~\ref{alg:computeProv}, it outputs the provenance of $q$
  w.r.t.\ $(\Pmc,\dataInstance)$.
\end{restatable}

%All the results presented in this section refer to the complexity obtained from considering the
%whole ontology as an input. In fact, more precise complexity analyses based on fixing some of the parameters
%are left for future work. Moreover, most of the queries used were simple instance or Boolean queries.
%In particular, this means that all the hardness results can be transferred directly to more complex (conjunctive)
%queries. On the other hand, i
The upper bounds from the previous section
for the general case obviously apply in the restricted idempotent case as well.

%%% Local Variables:
%%% mode: latex
%%% TeX-master: "paper"
%%% End:

% !TEX root =  paper.tex

\section{Evaluation}
\label{sec:evaluation}

To evaluate the feasibility of our approach, we implemented a prototype system
(\ontoprov) that extends the state-of-the-art OBDA system \ontop \cite{CCKK*17}
with the support for provenance.  \ontop supports SPARQL query answering over
ontologies in OWL\,2\,QL, the W3C standard corresponding to \DLLITE
\cite{W3Crec-OWL2-Profiles}.  The algorithm of \ontop has two stages, an
offline stage, which classifies the ontology and saturates the input set of
mappings, and an online stage, which rewrites the input queries according to
the saturated set of mappings. \ontoprov enriches these steps by taking into
account provenance information, and relies on ProvSQL \cite{2018-vldb-provsql}
to handle provenance from the database and queries in the mappings that go
beyond the CQ fragment.
% Ontop \cite{CCKK*17} is a state-of-the-art OBDA system that is applied
% in multiple practical scenarios, ranging from the electrical domain to
% the oil industries. This system works on the W3C counterpart of an
% OBDA scenario, where ontologies are expressed in
% OWL~2~QL\cite{W3Crec-OWL2-Profiles} and queries are expressed in
% SPARQL\cite{W3Crec-SPARQL-1.1-query}.
% We have extended Ontop with the ability of computing provenance. It is
% worth spending two lines on how this has been done. Ontop adopts a
% \emph{virtual} combined approach \cite{} to query answering, by
% relying on particular rewriting technique \cite{} which assumes that
% the ABox is \emph{saturated} with respect to the class and role
% inclusion axioms in the TBox. It is important to notice that such
% saturation is only \emph{virtual}, since it is achieved by compiling
% the ontology into the mappings rather than by materializing new ABox
% assertions. \dl{WELL, ACTUALLY IS ALMOST THE SAME} Due to the the
% approach adopted by Ontop, we cannot directly re-use the reformulation
% algorithm proposed here. To compute provenance, our implementation
% relies on named graph triple patterns and on
% OWL~2~QL ontologies, which are the W3C counterpart of DL-Lite To
% compute certain answers, Ontop uses a mixed strategy, namely a
% \emph{virtual} forward-chaining mixed with a particular rewriting
% technique called tree-witness rewriting.  The implementation is
% compliant with the approach presented here, but it differs
% substantially in the way provenance is computed.
%
We compare \ontop v3.0.0-beta-3 and \ontoprov over the BSBM \cite{BiSc09} and
the NPD \cite{LRXC15} benchmarks.  Experiments were run on a server with
2~Intel Xeon X5690 Processors (24 logical cores at 3.47\,GHz), 106\,GB of RAM
and five 1\,TB 15K\,RPM HDs. As RDBMS we have used PostgreSQL~11.2.

\mypar{Evaluation with the BSBM Benchmark} The BSBM benchmark is designed to test the different features of SPARQL. It provides a baseline for our tests, since it comes with an empty ontology and therefore it does not require ontological reasoning. In this experiment we restrict to a set of parametric queries (called here \emph{query mix}) in the benchmark (9 in total) that are supported by our theoretical framework.

Table~\ref{table:bsbm-exp} compares the average time (over three test runs) to evaluate the query mix with both \ontop and \ontoprov, on two datasets containing 10k and 1M products (resp., \textsl{bsbm10k} and \textsl{bsbm1M}). Evaluation times for both systems are very close. Hence, 
without a complex ontology or complex mappings, the overhead for computing provenance is rather small.

\begin{table} %[H]
  \footnotesize
  \centering
  %\vspace*{-2mm}
  \begin{tabular}{@{}l c c@{}}
    \toprule
    \textsl{Dataset} & \textsl{mixTime} \ontop & \textsl{mixTime} \ontoprov \\
    \midrule
    \textsl{bsbm10k} & 2.0s & 3.2s \\
    \textsl{bsbm1M} & 326s & 364s \\
    \bottomrule
  \end{tabular}
  \caption{BSBM Experiment}\label{table:bsbm-exp}
  \vspace{-1em}
\end{table}
% Much time to produce the uuids for provenance already with 1M products.

\mypar{Evaluation with the NPD Benchmark} As opposed to BSBM, the NPD Benchmark
is specifically tailored to OBDA systems; it comes with a complex ontology,
complex mappings, and queries of various kinds.
We restrict to 12 user queries that are supported by our framework. We use the
dataset \textsl{NPD}, containing real-world data about the oil extraction
domain, and the dataset \textsl{NPD10}, which is 10 times the size of
\textsl{NPD} and is generated by a \emph{data scaler}
\cite{DBLP:journals/semweb/LantiXC19}.
% As a technical note, we have slightly modified three of such queries by
% removing the date datatype, since \ontop sometimes failed in correctly
% unfolding those.
Differently from the BSBM benchmark, in NPD we observed many timeouts (set to
40 minutes) when running the benchmark queries with \ontoprov. This is due to
the fact that, in NPD, the optimizations performed by \ontop over the query
unfoldings are crucial for getting reasonably compact SQL queries.  Such
optimizations, however, need to be disabled in \ontoprov to guarantee
completeness. In fact, we are interested in \emph{all} the possible ways to
derive a result, and cannot identify and discard redundant derivations. For a
broader discussion about these aspects, please refer to the additional
material.

We assume that a user of \ontoprov would be more interested in understanding
the reason for a \emph{specific} answer tuple, rather than getting in bulk all
the possible explanations for all possible answer tuples. To simulate such user
interaction, in our tests we have actually instantiated the NPD queries with
answer tuples, and ran the obtained \emph{instantiated queries} (which are, in
fact, BCQs) over \ontoprov. Table~\ref{table:npd-exp} contains the aggregate
results of our runs. For each of our tests, we performed 5 test runs.
\begin{table}[t]
  \footnotesize
  \centering
  %\vspace*{-2mm}
  \begin{tabular}{@{}l@{\hspace{.5ex}}c@{\hspace{.5ex}}c@{\hspace{.5ex}}c | c@{\hspace{.5ex}}c@{\hspace{.5ex}}c@{\hspace{.5ex}}c@{\hspace{.5ex}}c@{\hspace{.5ex}}c@{}}
    \toprule
    \multicolumn{4}{c}{\ontop} & \multicolumn{6}{c}{\ontoprov} \\
    \textsl{Q} & \textsl{\#unf} & \textsl{t} & \textsl{t10} & \textsl{\#unf} & \textsl{\#uinst} & \textsl{tinst} & \textsl{tinst10} & \textsl{\#prov} & \textsl{\#prov10}\\
    \midrule
    1 & 0 & 29.5 & 172.4 & 16 & 16 & 5.2 & 1.7 & 1524 & 49\\
    2 & 0 & .5 & 3.1 & 32 & 32 &  .3 & .4 & 28 & 60\\
    3 & 24 & 5.0 & 51.1 & 16 & 16  & 248.0 & 296.5 & 84 & 153\\
    4 & 0 & 3.2 & 24.3 & 16 & 16  & 437.6 & 297.3 & 90 & 53\\
    5 & 0 & .1 & .2 & 0 & 0  & .1 & 1.5 & 1 & 1\\
    6 & 13 & 107.5 & 804.3 & 369  & 369  & 1439.3 & tout & 426 & tout\\
    7 & 0 & .4 & .2 & 0 & 0 & .1 & .3 & 1 & 1\\
    9 & 15 & 6.8 & 53.4 & 64 & 55 & .3 & .9 & 5 & 5\\
    10 & 1 & .4 & 25.4 & 4 & 4 & .8 & 11.3 & 1 & 7\\
    11 & 6 & 53.6 & 760.8 & 184 & 184  & 1342.6 & tout & 474 & tout\\
    12 & 8 & 69.1 & 1215.5 & 185 & 185  & 1248.1 & tout & 476 & tout\\
    31 & 21 & 60.3 & 633.0 & 248 & 239 & 1.9 & 5.3 & 120 & 60\\
    \bottomrule
  \end{tabular}
  \caption{NPD Experiment (times in seconds)}\label{table:npd-exp}
  \vspace{-1em}
\end{table}

The columns \textsl{\#unf} and \textsl{\#uinst} denote the number of times a
\texttt{UNION} operator appears in the unfolding of an NPD query and an
instantiated query, respectively. This measure gives an idea on the complexity
of the unfolding, and we can observe that the unfoldings produced by \ontoprov
are much more complex than those produced by \ontop. As argued above, this is
because \ontoprov disallows some optimizations. Columns \textsl{t} and
\textsl{t10} denote the average execution times of the queries over the
datasets \textsl{NPD} and \textsl{NPD10}, respectively, and for instantiated
queries these values are respectively denoted by \textsl{tinst} and
\textsl{tinst10}.  The execution times for \ontoprov are generally much higher
than for \ontop. We attribute this to the increased complexity of the
unfoldings.  Columns \textsl{\#prov} and \textsl{\#prov10} denote the number of
results for the instantiated queries, respectively over \textsl{NPD} and
\textsl{NPD10}. These numbers can be interpreted as the number of possible ways
an answer tuple can be derived, and give an indication on the complexity of
the benchmark itself. For instance, for query~\textsl{1} over the \textsl{NPD}
dataset there are on average 1524 explanations for a single answer tuple.

This test shows that the approach is feasible even with complex ontologies and mappings, 
but also that more work is needed in order to devise optimization techniques dedicated to a setting with provenance.

% !TEX root =  paper.tex

%~ \section{Related Work}
%~ %A key difference between the 
%~ %\cite{DBLP:journals/jair/CalvaneseOSS13,DBLP:conf/ekaw/CroceL18}.
%~ In addition to the already mentioned works on provenance semirings~\cite{bourgauxozaki,DBLP:journals/sigmod/Senellart17,DBLP:journals/ws/ZimmermannLPS12,Green07-provenance-seminal,GreenT17}, 
%~ we highlight some works on %the difference between our approach and 
%~ %existing works on 
%~ query explanation and axiom pinpointing. 

%~ \mypar{Query Explanation}
%~ and explanations~\cite{DBLP:journals/jair/CalvaneseOSS13,DBLP:conf/ekaw/CroceL18}

%~ \mypar{Axiom Pinpointing} Despite   apparent similarities, the problem of provenance is very
%~ different from that of axiom pinpointing~\cite{ScCo-03,parsia-iswc07,baader-ki07}. Although in  
%~ all these problems we are interested in tracing the causes of a consequence (or the lack of it), axiom pinpointing focuses on those that
%~ use \emph{minimal} sets of axioms. In contrast,  all possible derivations are relevant for provenance,
%~ independently of whether a subset of axioms might already suffice. Hence, although the set of monotone
%~ Boolean formulas is a semiring, the provenance polynomial over this semiring does not necessarily
%~ coincide with the pinpointing formula~\cite{BaPe10,BaPe10b}.
%~ %
%~ The main reason for this difference is that the provenance label represents a key which can 
%~ be used to encode other information or values, such as trust or probability; hence
%~ each derivation contributes to the final value represented 
%~ by the polynomial. 

\section{Conclusions and Discussion}

%\dl{Ana, Raf, Diego: Write this pls}%
We investigated the problem of dealing with provenance within  OBDA, 
%following the approach of representing provenance with semirings. 
%for databases
based on the provenance 
semiring approach introduced for databases. 
In our case, 
 every element of an OBDA instance
is annotated with provenance information. 
%labelled with a variable that allows the provenance of an answer to be identified.
%
We showed that query rewriting techniques can be applied to deal with provenance as well. An evaluation
based on a prototypical implementation shows that our methods are feasible in practice.

A key difference between the problem of provenance computation (or its decision version)
and that of axiom pinpointing~\cite{ScCo-03,parsia-iswc07,baader-ki07}
and query explanation~\cite{DBLP:journals/jair/CalvaneseOSS13,DBLP:conf/ekaw/CroceL18,DBLP:journals/jair/BienvenuBG19} %\todo{add more references}
is that
%, although in  
%all these problems we are interested in tracing the causes of a consequence (or the lack of it), 
axiom pinpointing 
and query explanation focus  on 
tracing the \emph{minimal} causes of a consequence (or the lack of it).
%those that
%use 
%that are . % sets of axioms. 
In contrast,  all possible derivations are relevant for provenance,
independently of whether a cause is minimal or not. 
%~ Although the set of monotone
%~ Boolean formulas is a semiring, the provenance polynomial over this semiring does not necessarily
%~ coincide with the pinpointing formula~\cite{BaPe10,BaPe10b}.
%~ %
%~ The main reason for this difference is that the provenance label represents a key which can 
%~ be used to encode other information or values, such as trust or probability; hence
%~ each derivation contributes to the final value represented 
%~ by the polynomial. 

As future work, we plan to investigate provenance with the monus operator.
We will also study the provenance of SPARQL query answering~\cite{GeertsKCF-icdt13} in OBDA. %\todo{GX: added}
Our implementation computes the provenance of a query assuming that the 
semiring is multiplicatively idempotent. While this assumption 
is useful to identify which parts of the knowledge base contribute 
to the query result, it restricts the applicability 
of our approach to other settings, in particular, to the numerical ones. 
For capturing  probabilities, 
it is important to distinguish repetitions, so (multiplicative)
idempotency is not suitable. 
In our setting, dropping 
the idempotency condition leads to cases where the polynomial can be infinite. 
It would be interesting to investigate whether the polynomial 
can be finitely represented, so that its computation 
could be applied in a numerical setting.
%, so that they can be computed and 
%used for numerical settings. 

\bibliographystyle{named}
%\bibliography{string-medium,references,krdb,w3c,local}
\bibliography{paper-bib}

\iftechrep
\clearpage
\appendix
% !TEX root =  paper.tex
\section{Proofs for Section~\ref{sec:prelim}}

\PropositionFinite*
\begin{proof}[Sketch]
Under multiplicative-idempotency, for any OBDA instance $(\obdaSpecification,\database)$
  and BCQ $q$, the number of possible monomials occurring $p\in\NPr$ such that
$(\obdaSpecification,\database)\models (q,p)$ is finite.
Thus, the only possibility for this set to
be infinite is if the   monomials   repeat  an unlimited number
of times. %That is, elements of the domain of an interpretation would
%have the same label.

To entail such  polynomials, arbitrarily many
repetitions should happen in all models of $(\obdaSpecification,\database)$.
%We conclude this proposition by arguing that,
However, under
multiplicative-idempotency,
any OBDA instance (based on an annotated \DLLITE ontology)
 enjoys the finite domain property~\cite{ArtEtAl1}.
That is, if an OBDA instance has a model then it has a model \Imc
with $\Delta^\Imc$ finite.
% (but $\Delta^\Imc_{\sf p}$ is infinite
%by definition of \NPr and semantics of \DLLITEp).
\end{proof}

\section{Proofs for Section~\ref{sec:rewriting}}

To show Lemma~\ref{lem:not-equal-marked} we use the following notation.
An OBDA instance $(\obdaSpecification_m,\database_m)$
is \emph{\marked for an OBDA instance} %if there is
$(\obdaSpecification,\database)$
%such that $(\obdaSpecification_m,\database_m)$
if it
is the result of: %applying the folllotransformation
%just described.
%Given a \DLLITEp ontology \Omc, we define $\Omc_m$ as the result of:
\begin{enumerate}
\item
replacing each $\pair{R(a,b)}{p}$ by $\pair{R_{a,b}(a,b)}{p}$, where
$R_{a,b}$ is a fresh role name;
\item for all $a,b\in\NI$ and all $R\in\NR$
occurring in $(\obdaSpecification,\database)$, adding a concept inclusion
$\pair{\exists R^{(-)}_{a,b}\sqsubseteq C}{v}$ for each
$\pair{\exists R^{(-)}\sqsubseteq C}{u}$
occurring in it, where $v\in\NV$;
\item
for all $a,b\in\NI$ and all $R,S\in\NR$
occurring in $(\obdaSpecification,\database)$,
adding a role inclusion $\pair{R^{(-)}_{a,b}\sqsubseteq S^{(-)}_{a,b}}{v}$ for each
$\pair{R^{(-)}\sqsubseteq S^{(-)}}{u}$
occurring in it,
where $v\in\NV$;
\item replacing the annotation of each axiom, mapping
and tuple in a relation of \dataInstance by fresh
$v\in\NV$ (including the axioms of items above),
so that each $v$ is unique.
\end{enumerate}
%
%~ Given a function $\dagger:\NV\rightarrow\NV$, called \NV-function,
%~ and an OBDA instance $(\obdaSpecification_m,
%~ \database_m)$, we write $(\obdaSpecification^\dagger_m,
%~ \database^\dagger_m)$ for the result of  simultaneously replacing
%~ each variable $p\in \NV$ occurring in $(\obdaSpecification_m,
%~ \database_m)$
%~ by $\dagger(p)$, and similarly for annotated queries $(q_m,p_m)$.
%~ %
%~ The next lemma establishes the connection between
%~ an OBDA instance $(\obdaSpecification,\database)$
%~ and a translation into a marked OBDA instance $(\obdaSpecification_m,\database_m)$.
Intuitively, we want to ensure that there is a model of $(\obdaSpecification_m,\database_m)$ where
elements in the anonymous part (i.e., not in the image of \NI)
connected (via roles) to the image of an individual
are associated with monomials containing at least one variable of the
semiring, which is not shared by  anonymous elements connected
to  the image of another individual. In other words,
we want to `mark' monomials associated to
elements derived from assertions of named individuals.

Conditions~$1$--$4$
%that (a) the annotations are unique and
%(b) the OBDA instance $(\obdaSpecification,\database)$ is
% \marked
are necessary and sufficient
to ensure that we cannot find
two monomials which are mathematically equal in $\Imc_{(\obdaSpecification,\database)}$.
Clearly, if Condition~$4$ does not hold we may find monomials mathematically equal in
$\Imc_{(\obdaSpecification,\database)}$.
As we show in the following example, this may also happen if $4$ holds
but not $1$--$3$.
%we may find monomials mathematically equal in
%$\Imc_{(\obdaSpecification,\database)}$.
\begin{example}\label{ex:anonymous}
Let $((\Omc,\mapping,\Smc),\database)$ be an OBDA instance
with
\[
  \begin{array}{l@{\qquad}l}
    (\exists R\sqsubseteq \exists S,u)
    & (\exists R^-\sqsubseteq \exists S^-,r)\\
    (\exists S\sqsubseteq \exists R^-,s)
    & (\exists S^-\sqsubseteq \exists R,t)
  \end{array}
\]
in \Omc and $(R(a,b),p)\in \mapping(\database)$.
Then  there are two tuples in $R^{\Imc_{(\obdaSpecification,\database)}}$
with the annotation $p\times r \times s\times u\times t$.
% has
% $(d,e,p\times q \times t\times r\times s)$ and
% $(d',e',p\times r \times s\times q\times t)$, for some
% $d,e,d',e'\in \Delta^{\Imc_{(\obdaSpecification,\database)}}$.
\end{example}

To show Lemma~\ref{lem:not-equal-marked}  we use the classical notion of
a \emph{canonical model}  of an OBDA instance.
As in Section~\ref{sec:rewriting}, %to simplify the presentation,
for each role $R^-$ in the OBDA instance
$((\Omc,\mapping,\Smc),\dataInstance)$, we extend \Omc with
the   role inclusions $(R^-\sqsubseteq \overline{R},p_R)$
 and $(\overline{R}\sqsubseteq R^-,p'_R)$, where $\overline{R}$ is a fresh role name and
 $p_R,p'_R$ are fresh variables of a provenance semiring.
Assume w.l.o.g.\ that inverse roles only occur in such equivalences.
Let $\NM_{\sf min}$ be a \emph{min}imal subset of $\NM$ such that
for all $p\in\NM$ there is $q\in\NM_{\sf min}$ where
$p$ and $q$ are mathematically equal.
We define the canonical model $\Imc_{(\obdaSpecification,\database)}$ of a \marked OBDA instance
$(\obdaSpecification,\dataInstance)$, with $\obdaSpecification=(\Omc,\mapping,\Smc)$,
 as the union of $\Imc_0,\Imc_1,\ldots$, where the $\Imc_n$s are inductively defined as follows. For $n=0$,
$\Imc_0$ is defined by:
\[
  \begin{array}{rcl}
    \Delta^{\Imc_0} &=& \NI,\\
    \Delta^{\Imc_0}_{\sfm} &=& \NM_{\sf min},\\
    A^{\Imc_0} &=& \{ (a,p) \mid \pair{A(a)}{p}\in \mapping(\dataInstance)\},\\
    R^{\Imc_0} &=& \{ (a,b,p)\mid \pair{R(a,b)}{p}\in \mapping(\dataInstance)\},
  \end{array}
\]
for all $A\in \NC$ and all $R\in\NR$,    $a^{\Imc_0}=a\text{ for every }a\in\NI$
and $p^{\Imc_0}=q\in \NM_{\sf min}$, with $p$ and $q$ mathematically equal,  for every $p\in \NM$.
Assume now that $\Imc_{n}$ is defined.
%Given a positive inclusion $I\in\Omc$, we
% obtain
We define $\Imc_{n+1}$ by choosing a positive inclusion $I\in\Omc$
and applying one of the following rules, %
 %\rpn{this is not clear; it looks as if $\Imc_{n+1}$ is defined by one inclusion, but later speak about several rules.
 %I understand what you say, but we should write it better.} \ano{it is because we have one rule for each type of inclusion.
 %Not sure what you mean. }
\begin{itemize}[itemindent=0mm,leftmargin=2.5mm,labelsep=1mm]
\item if $I= \pair{A_1\sqsubseteq A}{p} $,
   $(a,v^{\Imc_n})\in A_1^{\Imc_n}$, and $\vec{t}=(a,(p\otimes v)^{\Imc_n})\not\in A^{\Imc_n}$,
then add $\vec{t}$ to $A^{\Imc_n}$,
\item if $I=  \pair{R_1\sqsubseteq R}{p} $,
   $(a,b,v^{\Imc_n})\in R_1^{\Imc_n}$, and $\vec{t}=(a,b,(p\otimes v)^{\Imc_n})\not\in R^{\Imc_n}$,
then add $\vec{t}$ to $R^{\Imc_n}$,
\item if $I=  \pair{R_1 \sqsubseteq R^-}{p} $ or $I=\pair{R_1^- \sqsubseteq R}{p}$,
   $(a,b,v^{\Imc_n})\in R_1^{\Imc_n}$, and $\vec{t}=(b,a,(p\otimes v)^{\Imc_n})\not\in R^{\Imc_n}$,
then add $\vec{t}$ to $R^{\Imc_n}$,
\item if $I=  \pair{\exists R\sqsubseteq A}{p} $,
  there is $b$ such that $(a,b,v^{\Imc_n})\in R^{\Imc_n}$, and
  $\vec{t}= (a,(p\otimes v)^{\Imc_n})\not\in A^{\Imc_n}$, then
add $\vec{t}$ to $A^{\Imc_n}$,
\item if $I=  \pair{A\sqsubseteq \exists R}{p} $,
$(a,v^{\Imc_n})\in A^{\Imc_n}$, and there is no $b$ such that $\vec{t}=(a,b,(p\otimes v)^{\Imc_n})\in R^{\Imc_n}$
then add a fresh element $b$ to $\Delta^{\Imc_n}$ and add $\vec{t}$ to $R^{\Imc_n}$,
\item if $I=  \pair{\exists R_1\sqsubseteq \exists R}{p} $,
there is $b$ such that $(a,b,v^{\Imc_n})\in R_1^{\Imc_n}$, and there is no $c$ such
that $\vec{t}=(a,c,(p\otimes v)^{\Imc_n})\in R^{\Imc_n}$
then add a fresh element $c$ to $\Delta^{\Imc_n}$, add $\vec{t}$ to $R^{\Imc_n}$.
\end{itemize}
We assume that rule application is fair, i.e., if a rule is applicable at a certain place,
it will   eventually be applied there. %Also, assume $x^{\Imc_n}=x\text{ for every }x\in\NI\cup\NPr$.
$\Imc_{(\obdaSpecification,\database)}$ is defined as the union of all such $\Imc_n$,
where $x^{\Imc_{(\obdaSpecification,\database)}}=x^{\Imc_0}\text{ for every }x\in\NI\cup\NM$.

Importantly, the last element of each tuple in $\Imc_{(\obdaSpecification,\database)}$
is  mathematically distinct from the others. This holds for
$\Imc_0$ since each axiom of \Omc is associated
with a variable of a semiring
appearing in at most one axiom.
For $n>0$, we have to consider tuples created using fresh anonymous elements.
The interpretation of each individual from $\NI$
occurring in \Omc can be connected via roles to an anonymous part of $\Imc_{(\obdaSpecification,\database)}$.
We propagate the annotations (which are unique) associated to
an individual to this anonymous part. Clearly, if these individuals
are not connected, annotations associated to them form disjoint sets
and the monomials are   mathematically distinct.
For connected individuals, we use the assumption that \Omc is
\marked.
%If the individuals are connected, e.g., $\pair{R_{a,b}(a,b)}{p}\in\Omc$
%then
Our assertions are of the form $\pair{R_{a,b}(a,b)}{p}$
and since we add only
concept and role inclusions of the form
 $\pair{\exists R^{(-)}_{a,b}\sqsubseteq C}{v}$ and
 $\pair{ R^{(-)}_{a,b}\sqsubseteq S^{(-)}_{a,b}}{v}$,
 the extension of the fresh roles $R_{a,b}$
can only have elements in the image of named individuals
and the extension  of $R$ can only have tuples where
at least one element is anonymous. So $R^{\Imc_{(\obdaSpecification,\database)}}_{a,b}
\cap R^{\Imc_{(\obdaSpecification,\database)}}=\emptyset$.

With this separation, if an anonymous element is connected to $a$ after
applying, e.g., an inclusion $\pair{\exists R_{a,b}\sqsubseteq \exists S}{p}$
in the construction of $\Imc_{(\obdaSpecification,\database)}$, then we know that $p$ must occur in
the monomials associated to tuples containing this anonymous individual.
On the other hand, if another anonymous element is connected to $b$ after
applying an inclusion $\pair{\exists R^-_{a,b}\sqsubseteq \exists S}{v}$
in the construction of $\Imc_{(\obdaSpecification,\database)}$, then we know that now $v$ must occur in
the monomials  associated to tuples containing this other anonymous individual.
We never re-apply these inclusions containing fresh roles. So $p$ and $v$
 mark the anonymous part of $\Imc_{(\obdaSpecification,\database)}$ connected to
$a$ and $b$, respectively.
%then $p$ may occur in monomials in the anonymous parts  connected  to
%$a$ and $b$. Though, we can only derive  from $\pair{R(a,b)}{p}$
%an anonymous individual connected to $a$ with an inclusion of the form
%$\pair{\exists R\sqsubseteq C}{v}$,
% and  we can only derive  from $\pair{R(a,b)}{p}$
%an anonymous individual connected to $b$ with an inclusion of the form
% $\pair{\exists R^-\sqsubseteq D}{u}$.
By definition of $(\obdaSpecification,\dataInstance)$, $v$ and $u$ are distinct variables of a semiring,
this means that the last element of each tuple in $\Imc_{(\obdaSpecification,\database)}$
is  mathematically distinct from the others.

%~ To show Theorem~\ref{thm:trans} we use the following technical lemma.
%~ whose proof uses the construction of the canonical model and the assumption
%~ that the OBDA instance is in \marked form.
%that the annotation of each axiom in a \DLLITEp ontology is unique.

%~ \begin{restatable}{lemma}{Lemmanotequal}\label{lem:not-equal}
%~ Let \Omc be a satisfiable \DLLITEp \marked ontology,
%~ $q$ a BCQ and $p$ a polynomial in \NPr.
%~ If $\Omc\models (q,p)$ then for any two monomials $p_1,p_2$ appearing in $p$, it holds that $p_1$ and $p_2$
%~ are mathematically distinct.
%~ \end{restatable}

%\Lemmanotequal*

%To show Lemma~\ref{lem:not-equal-marked} we use the notion of
%a homomorphism and the following technical lemma. %Lemma~\ref{lem:not-equal}.
%\begin{definition}
Let $\Imc,\Jmc$ be annotated interpretations.
A \emph{homomorphism}
 is a function $h:\Imc\rightarrow\Jmc$ from $\Delta^\Imc$ to $\Delta^\Jmc$
such that:
\begin{itemize}[itemindent=0mm,leftmargin=2.5mm,labelsep=1mm]
\item $h(a^\Imc)=a^\Jmc$ for all $a\in \NI\cup\NM$;
%~ \item $(a^\Imc,p^\Imc)\in A^\Imc$
%~ implies $(h(a^\Jmc),(h(p^\Jmc))\in A^\Jmc$, for all $A\in\NC$, $a\in\NI$ and
%~ $p\in \NPr$;
and
\item  $(\vec{a},p^\Imc)\in E^\Imc$
implies $(h(\vec{a}),h(p^\Jmc))\in E^\Jmc$, for all $E\in\NC\cup\NR$ and
$p\in\NM$;
\end{itemize}
where $\vec{a}=a^\Imc$, if $E\in\NC$, and
%$h(\vec{a})$ is a shorthand for $h(a)$ if $E\in\NC$ and $\vec{a}=a$, and
$\vec{a}=(a^\Imc,a'^{\Imc})$,
with $h(\vec{a})=(h(a^\Imc),h(a'^\Imc))$, if $E\in\NR$,
for $a,a'\in\NI$.
We write $\Imc\rightarrow\Jmc$ if there is
a  homomorphism from \Imc to \Jmc.
%\end{definition}
\begin{restatable}{lemma}{Lemmanotequal}\label{lem:not-equal}
Let $(\obdaSpecification,\dataInstance)$ be a satisfiable \marked OBDA
instance (for some OBDA instance),
$q$ a BCQ and $p\in\NPr$ a polynomial.
%~ If $(\obdaSpecification,\dataInstance)\models (q,p)$ then for
%~ any two monomials $p_1,p_2$   in $p$, it holds that $p_1$ and $p_2$
%~ are mathematically distinct.
%Moreover,
%assuming that all monomial in $p$ are  mathematically distinct,
%we have that
$(\obdaSpecification,\dataInstance)\models (q,p)$
if, and only if, $\Imc_{(\obdaSpecification,\dataInstance)}\models (q,p)$.
\end{restatable}
\begin{proof}
%~ We now argue that $(\obdaSpecification,\dataInstance)\models (q,p)$
%~ if, and only if, $\Imc_{(\obdaSpecification,\dataInstance)}\models (q,p)$.
By assumption, $(\obdaSpecification,\dataInstance)$ is satisfiable and so,
by construction, $\Imc_{(\obdaSpecification,\dataInstance)}$ is a model
of $(\obdaSpecification,\dataInstance)$.
So  $(\obdaSpecification,\dataInstance)\models (q,p)$ implies
 $\Imc_{(\obdaSpecification,\dataInstance)}\models (q,p)$.
 The converse direction follows from the standard notion that
 if an interpretation \Jmc (here annotated) satisfies
 $(\obdaSpecification,\dataInstance)$ then there is a homomorphism %\ano{def homomorphism}
    $\Imc_{(\obdaSpecification,\dataInstance)}\rightarrow\Jmc$.
 Since the last element of each tuple in $\Imc_{(\obdaSpecification,\database)}$
is  mathematically distinct from the others
     $\Imc_{(\obdaSpecification,\dataInstance)}\rightarrow\Jmc$
can only be injective. So not only $\Jmc\models q$ but there is also a
$1$ to $1$ correspondence between
the   matches of $q$ in $\Imc_{(\obdaSpecification,\database)}$, with the respective annotations,
and the matches of $q$ in \Jmc.
Thus, if $\Imc_{(\obdaSpecification,\database)}\models (q,p)$
then $\Jmc\models (q,p)$. As \Jmc is an arbitrary annotated intepretation
satisfying $(\obdaSpecification,\database)$, it follows that
$(\obdaSpecification,\dataInstance)\models (q,p)$.
\end{proof}

%~ \begin{restatable}{lemma}{Lemmamarkedtranslation}\label{lem:markedtranslatio}
%~ Given a satisfiable annotated OBDA instance
%~ $(\obdaSpecification,\dataInstance)$ and an annotated query $(q,p)$,
%~ there is  %$\dagger$,
 %~ $(\obdaSpecification_m,\dataInstance_m)$ and
   %~ $(q_m,p_m)$
    %~ such that
%~ \begin{itemize}
%~ %\item $(\obdaSpecification,\database)\subseteq(\obdaSpecification^\dagger_m,
%~ %\database^\dagger_m)$, $(q,p)\subseteq (q^\dagger_m,p^\dagger_m)$;
%~ \item $(\obdaSpecification,\dataInstance)\models (q,p)$ iff
%~ $(\obdaSpecification_m,\dataInstance_m)\models (q_m,p_m)$;
%~ \item $(\obdaSpecification_m,\dataInstance_m)$ is marked for $(\obdaSpecification,\dataInstance)$; and
%~ \item  $|(\obdaSpecification_m, \database_m)|+|(q_m,p_m)|$ is polynomially bounded by
%~ $|(\obdaSpecification,\database)|+|(q,p)|$.
%~ \end{itemize}
%~ \end{restatable}

We are now ready to prove Lemma~~\ref{lem:not-equal-marked}.

\Lemmanotequalmarked*

%~ Such unique `marks'   allow us to ensure that
%~ the same monomial cannot be obtained twice
%~ and, so, our translation in Step 2 does not require
%~ inequalities to distinguish $p+p$ from $p$ (because $p+p$ is never entailed). This is crucial for the query entailment problem
%~ since entailment of conjunctive queries with inequalities
%~ in \DLLite is undecidable~\cite{Guti}.

%\Lemmamarkedtranslation*
\begin{proof}
We first argue that if $(\obdaSpecification_m,\database_m)$
is \marked for some OBDA instance $(\obdaSpecification,\database)$ then
$(\obdaSpecification_m,\database_m)$ only entails
annotated queries $(q_m,p_m)$ such that
any two monomials $p_1,p_2$ appearing in $p_m$ %, it holds that $p_1$ and $p_2$
are mathematically distinct.
By assumption $(\obdaSpecification,\database)$ is satisfiable,
 so $(\obdaSpecification_m,\database_m)$ is satisfiable and, moreover,
$\Imc_{(\obdaSpecification,\database)}$ is a model of $(\obdaSpecification,\database)$.
Assume that $p\in\NPr$ contains two monomials which are mathematically equal.
Then, for any BCQ $q$, $(\obdaSpecification,\database)\models (q,p)$ iff
$(\obdaSpecification,\database)\models q$ and for every model \Imc of $(\obdaSpecification,\database)$,
$p\subseteq \p{\Imc}{q}$. This means that for each occurrence of a monomial in $p$
we find an occurrence of it in $\p{\Imc}{q}$.
However,  for $\Imc=\Imc_{(\obdaSpecification,\database)}$, we know that we cannot find
two monomials which are mathematically equal in $\p{\Imc}{q}$.

Now we argue about the second point of this lemma.
%that $(\obdaSpecification,\dataInstance)\models (q,p)$ iff
%$(\obdaSpecification_m,\dataInstance_m)\models (q_m,p_m)$.
Let $(\obdaSpecification_m,\database_m)$
be a \marked OBDA instance for $(\obdaSpecification,\database)$
and let $\dagger:\NV\rightarrow\NV$ be the function
that maps $v\in\NV$ occurring in $(\obdaSpecification_m,\database_m)$
to $\dagger(v)$ in $(\obdaSpecification,\database)$.
%Given a function $\dagger:\NV\rightarrow\NV$, called \NV-function,
%and an OBDA instance $(\obdaSpecification_m,
%\database_m)$,
 Given $p_m\in\NPr$ and a function $\dagger:\NV\rightarrow\NV$,
$p^\dagger_m$
is the result of simultaneously replacing each occurrence of
$v\in\NV$  in $p_m$ by $\dagger(v)$.
Similarly, given a query $q_m$ and
%a polynomial $p_m$,
$q^\dagger_m$
%and $p^\dagger_m$
is the result of simultaneously replacing each occurrence of
$v\in\NV$  in $q_m $
%and $p_m$
by $\dagger(v)$ and, in addition, we   replace each
$R_{a,b}$ by $R$, where $R\in\NR$ and $a,b\in\NI$.
%and $p_m$
%, respectively.
%Moreover, we replace each $R_{a,b}$ by $R$
We  use $\dagger$ to define a mapping between interpretations.
%satisfying $(\obdaSpecification_m,\database_m)$  and
%$(\obdaSpecification,\database)$.
Given %a function $\dagger:\NV\rightarrow\NV$  and
 an annotated interpretation $\Imc=(\Delta^\Imc,\Delta^\Imc_\sfm,\cdot^\Imc)$,
we
%choose a representative
%$p_{\sf rep}\in \{p'\mid \dagger(p')=p\}$,
%for every $p\in\NPr$, and
define
   $\Imc^\dagger$ as %the annotated interpretation
$(\Delta^{\Imc},\Delta^{\Imc}_\sfm,\cdot^{\Imc^\dagger})$
with %$\Delta^{\Imc^\dagger}=\Delta^\Imc$, $\Delta^{\Imc^\dagger}_\sfm=\Delta^\Imc_\sfm$
%and
$\cdot^{\Imc^\dagger}$ satisfying:
\begin{itemize}
\item $a^{\Imc_\dagger}=a^\Imc\in\Delta^\Imc$;
\item $p^{\Imc_\dagger}=(p^\dagger)^\Imc \in\Delta^\Imc_\sfm$;
%~\footnote{Under
%the associative, commutative and distributive laws. E.g., $((p\oplus q) \otimes r)^\Imc = ((p\otimes r)\oplus (q\otimes r))^\Imc$
%by distributivity.};
\item $A^{\Imc^\dagger}=\{(d,(p^\dagger)^\Imc)\mid
(d,p^\Imc)\in A^{\Imc} \}$;
\item $R^{\Imc^\dagger}= \
  \begin{array}[t]{@{}l}
    \{(d,e,(p^\dagger)^\Imc)\mid (d,e,p^\Imc)\in R^{\Imc}\} \cup{}\\
    \{(a^\Imc,b^\Imc,(p^\dagger)^\Imc)\mid (a^\Imc,b^\Imc,p^\Imc)\in
    R_{a,b}^\Imc\},
  \end{array}$
  % to some $R^\Imc\subseteq \Delta^\Imc\times\Delta^\Imc\times
  % \Delta^\Imc_\sfm$.
\end{itemize}
for every $a,b\in\NI$, $p\in \NM$, $A\in\NC$ and $R\in\NR$.  The following
claim can be proved by structural induction.

\begin{claim}\label{cl:one}
Let $(\obdaSpecification_m,\database_m)$
be a \marked OBDA instance for $(\obdaSpecification,\database)$
and let $\dagger:\NV\rightarrow\NV$ be the function
that maps $v\in\NV$ occurring in $(\obdaSpecification_m,\database_m)$
to $\dagger(v)$ in $(\obdaSpecification,\database)$.
For every annotated interpretation \Imc, the following holds:
\begin{itemize}
\item if $\Imc\models (\obdaSpecification_m,\database_m)$ then
$\Imc^\dagger\models (\obdaSpecification,\database)$; and
\item if $\Imc\models (\obdaSpecification,\database)$ then there is
$\Jmc$ such that $\Imc=\Jmc^\dagger$ and
 $\Jmc\models (\obdaSpecification_m,\database_m)$.
\end{itemize}
\end{claim}
We first show that %, for any    \marked OBDA instance $(\obdaSpecification_m,\database_m)$,
 if $(\obdaSpecification_m,\database_m)\models (q_m,p_m)$
 and $(q^\dagger_m,p^\dagger_m)=(q,p)$ hold,
 then
 $(\obdaSpecification,\database)\models (q,p)$.
 Assume that $(\obdaSpecification_m,\database_m)\models (q_m,p_m)$ and $(q^\dagger_m,p^\dagger_m)=(q,p)$.
 % and
%there is a function $\dagger:\NV\rightarrow\NV$ such that
%$(\obdaSpecification,\database)=(\obdaSpecification^\dagger_m,
%\database^\dagger_m)$, $(q,p)=(q^\dagger_m,p^\dagger_m)$.
If $\Imc\models (\obdaSpecification,\database)$ then,
 by Claim~\ref{cl:one}, there is
$\Jmc$ such that $\Imc=\Jmc^\dagger$ and
 $\Jmc\models (\obdaSpecification_m,\database_m)$.
 As $(\obdaSpecification_m,\database_m)\models (q_m,p_m)$, we have that
 $\Imc_{(\obdaSpecification_m,\database_m)}\models(q_m,p_m)$
 and there is a homomorphism from $\Imc_{(\obdaSpecification_m,\database_m)}$
 to \Jmc. By definition of $(\obdaSpecification_m,\database_m)$,
 a tuple (in a relation) can only be associated with multiple annotations
 in $\Imc_{(\obdaSpecification_m,\database_m)}$
 if this tuple is in $\Imc_0$, meaning that it is in the image
 of an   assertion in $\Mmc_m(\database_m)$, with $\Pmc_m=(\Omc_m,\Mmc_m,\Smc_m)$.
 Moreover, $(\obdaSpecification_m,\database_m)$ has the same number of
 assertions as in $(\obdaSpecification,\database)$, the only difference
 is the renaming of roles in Point 1, and the renaming of annotations in Point 4.
 Thus, if a tuple (in a relation) is associated with $k$ annotations
 in $\Imc_{0}$  then
 the corresponding assertion is also associated with $k$ annotations
 in $(\obdaSpecification,\database)$, where the mapping between annotations
 is given by $\dagger$. So
 $\Imc^\dagger_{(\obdaSpecification_m,\database_m)}$   respects
 the multiplicity of   assertions in  $\Imc_{(\obdaSpecification_m,\database_m)}$.
 Thus,
  $\Imc_{(\obdaSpecification_m,\database_m)}\models (q_m,p_m)$ implies
  $\Imc^\dagger_{(\obdaSpecification_m,\database_m)}\models (q^\dagger_m,p^\dagger_m)$.
 By definition of $\Imc^\dagger_{(\obdaSpecification_m,\database_m)}$,
 there is a homomorphism from $\Imc^\dagger_{(\obdaSpecification_m,\database_m)}$ to
 $\Jmc^\dagger$. Since $(q^\dagger_m,p^\dagger_m)=(q,p)$ and $\Jmc^\dagger=\Imc$
 we have that $\Imc\models (q,p)$. As \Imc
 is an arbitrary annotated interpretation satisfying
 $(\obdaSpecification,\database)$,
 we conclude that $(\obdaSpecification,\database)\models (q,p)$.

 %($\Rightarrow$)
 We now show that if $(\obdaSpecification,\database)\models (q,p)$ holds,
 then  $(\obdaSpecification_m,\database_m)\models (q_m,p_m)$,
 %
 %there is a \marked OBDA instance $(\obdaSpecification_m,\database_m)$
 %for $(\obdaSpecification,\database)$,
 for some     query $(q_m,p_m)$ and
  function $\dagger:\NV\rightarrow\NV$
  mapping $v\in\NV$   in $(\obdaSpecification_m,\database_m)$
to $\dagger(v)$ in $(\obdaSpecification,\database)$ and
  such that
% $(\obdaSpecification,\database)=(\obdaSpecification^\dagger_m,
%\database^\dagger_m)$,
 $(q,p)=(q^\dagger_m,p^\dagger_m)$.
%and.
%
 Assume $(\obdaSpecification,\database)\models (q,p)$.
In this lemma we assume that $(\obdaSpecification,\database)$
is satisfiable, and so, for any
marked
$(\obdaSpecification_m,\database_m)$ for  $(\obdaSpecification,\database)$,
it follows from the definition  that $(\obdaSpecification_m,\database_m)$ is satisfiable.
This means that
a canonical model $\Imc_{(\obdaSpecification_m,\database_m)}$
of $(\obdaSpecification_m,\database_m)$ satisfies
 $(\obdaSpecification_m,\database_m)$.
%we have that
% $\Imc_{(\obdaSpecification_m,\database_m)}\models(\obdaSpecification_m,\database_m)$.
Then, by Claim~\ref{cl:one},
%there is an \NV-function , and
$\Imc^\dagger_{(\obdaSpecification_m,\database_m)}\models(\obdaSpecification,\database)$.
As $(\obdaSpecification,\database)\models (q,p)$,
we have that $\Imc^\dagger_{(\obdaSpecification_m,\database_m)}\models(q,p)$.
The fact that  $\Imc_{(\obdaSpecification_m,\database_m)}\models(q_m,p_m)$ for some
 $(q_m,p_m)$
 such that $(q,p)=(q^\dagger_m,p^\dagger_m)$
 follows from the construction of  $\Imc^\dagger_{(\obdaSpecification_m,\database_m)}$.
Since $\Imc_{(\obdaSpecification_m,\database_m)}$
is a canonical model of $(\obdaSpecification_m,\database_m)$, by
Lemma~\ref{lem:not-equal},
 $(\obdaSpecification_m,\database_m)\models (q_m,p_m)$.

Finally, as $(\obdaSpecification_m,\database_m)$
is \marked for $(\obdaSpecification,\database)$, it is easy to see that
$|(\obdaSpecification_m,\database_m)|$
is polynomial in $|(\obdaSpecification,\database)|$, and therefore,
also in $|(\obdaSpecification,\database)|+|(q,p)|$.
Regarding the size of $(q_m,p_m)$,
one can easily see that it depends polynomially on
the size of $(q,p)$,
since $(q,p)$ is the result of replacing role
names and variables from \NV in  $(q_m,p_m)$.
Dependence on $|(\obdaSpecification,\database)|$ is due
to the   function $\dagger$ and
the requirement that annotations
are unique in $(\obdaSpecification_m,\database_m)$, which affects
%the fact that
the number of bits necessary to encode the
annotations   occurring in $(q_m,p_m)$.
% depends
%on $|(\obdaSpecification,\database)|$.
%(dependence on the
%size of $(\obdaSpecification,\database)$
%is due to the requirement that annotations
%are unique in $(\obdaSpecification_m,\database_m)$ , thus,
%the number of bits necessary to encode the
%annotations which occur in $(q_m,p_m)$ depends
%on $|(\obdaSpecification,\database)|$).
\end{proof}

%
%\ano{}
%~ \begin{example}
%~ Lemma~\ref{lem:not-equal} does not hold for \DLLITEp ontologies in general.
%~ To see this consider $\Omc=\{\pair{R(a,b)}{p},
%~ \pair{\exists R\sqsubseteq \exists R}{v},
%~ \pair{\exists R^-\sqsubseteq \exists R^-}{u}, \pair{A\sqsubseteq \neg B}{w}\}$.
%~ Then, $\Omc\models \exists$
%~ \end{example}

With the help of this lemma, we can state the main result of this section; namely, that the translation
for marked ontologies is correct.
\Theoremtranslation*
\begin{proof}
By definition of entailment from an ontology, it suffices to show that,
for every annotated interpretation \Imc,
\begin{center}
%\hfil
 $\Imc\models (q,p)$ iff $\Imc\models \tr{q}{p}$. \hfill  ($*$)
\end{center}
Indeed, if ($*$) holds, then $(\obdaSpecification,\database)\models (q,p)$ iff for every model
\Imc of $(\obdaSpecification,\database)$ we have $\Imc\models (q,p)$
iff for every model \Imc of $(\obdaSpecification,\database)$ we have $\Imc\models \tr{q}{p}$
iff $(\obdaSpecification,\database)\models \tr{q}{p}$.

We now show the claim ($*$). Let \Imc be an arbitrary annotated interpretation.
%\begin{claim}\label{cl:char}
%Let \Imc be a \DLLITEp interpretation, $q$ a BCQ and $p$ a polynomial in \NPr.
%Then, $\Imc\models (q,p)$ iff
%$\Imc\models \tr{q}{p}$.
%\end{claim}
%\noindent
%\textit{Proof of the Claim.}

Assume first that $\Imc\models (q,p)$, and let $n$ be the number of monomials in $p$.
Then, $\Imc\models q$ and $p\subseteq \p{\Imc}{q}$.
By definition of $\p{\Imc}{q}$, there is a
set $\chi_\Imc(q)$ consisting of $n$   matches of $q$
in \Imc such that:
\begin{equation}\label{eq}
 p=\sum_{\pi\in\chi_\Imc(q)}\prod_{P(\vec{t},t)\in q} \map^-(t). %\p{\Imc}{P(\map(\vec{t}))}.
\end{equation}
%~ This means that, for each $\pi\in\chi_\Imc(q)$,   monomials in $\{\p{\Imc}{P(\map(\vec{t}))}\mid P(\vec{t})\in q\}$
%~ partition $p_\pi=\prod_{P(\vec{t})\in q} \p{\Imc}{P(\map(\vec{t}))}$; and
  %~ monomials in $\{p_\pi\mid \pi \in \chi_\Imc(q)\}$    decompose  $p$.
%
Consider first the case that all terms in $q$ are variables. By definition of $\tr{q}{p}$,
there is a function $\epsilon$ mapping $\chi_\Imc(q)$ to $\{1,\dots,n\}$
and a BCQ $q'$ of the form presented in Equation~\ref{eq:trans} such that
the last term $t$ of   atom $P(\vec{t},t)$
in $q_{\epsilon(\pi)}$ is any monomial $v\in\NM$ such that $v^\Imc=\map^-(t)$.
%$\p{\Imc}{P(\map(\vec{t}))}$.
%
Let $\pi'_{\epsilon(\pi)}$ be the result of replacing each $x\in\vec{x}$ in the domain
of $\pi_{\epsilon(\pi)}$ by $x_{\epsilon(\pi)}\in\vec{x}_{\epsilon(\pi)}$.
By definition of $q_{\epsilon(\pi)}$, each $\pi'_{\epsilon(\pi)}$ is
a match of $q_{\epsilon(\pi)}$ in \Imc.
Then, $\bigcup_{1\leq \epsilon(\pi)\leq n} \pi'_{\epsilon(\pi)}$
is a match of $q'$ in \Imc. So,  $\Imc\models \tr{q}{p}$.

Conversely, assume that $\Imc\models \tr{q}{p}$. Then there is a match $\map=\bigcup_{1\leq i\leq n} \map_i$
of some $q'\in\tr{q}{p}$ in \Imc, where each $\map_i$
is a match if the `copy' $q_i$ of $q$ in \Imc, $1\leq i\leq n$.
Let $\map'_i$ be the result
of replacing each variable $x_i\in\vec{x}_i$ by $x\in\vec{x}$,   $1\leq i\leq n$.
By definition of $q_i$ and $q$, each $\map'_i$ is a match of $q$ in \Imc.
Moreover, since  any two monomials in $p$ are   not mathematically equal,
  $\map'_1,\ldots,\map'_n$ are all distinct.
By definition of $q_i$, the product $p_i$ of the polynomials in the domain of $\map'_i$
is a monomial in $p$, and the sum of the monomials $p_1,\ldots,p_n$
is equal to $p$.
We then obtain the same equality of Equation~\ref{eq}.
Thus, $p\subseteq \p{\Imc}{q}$.
The proof for BCQs with individual names is similar, except that
`copies' of the query contain individual names in the corresponding
positions.
%
%This finishes the proof of the claim.
%
%\smallskip
%
%Assume $\Omc\models (q,p)$ and let $\Imc$ be an arbitrary
%\DLLITEp interpretation satisfying \Omc. Then, $\Imc\models (q,p)$
%and, by the Claim, $\Imc\models \tr{q}{p}$. Since
%$\Imc$ was arbitrary,  $\Omc\models (q,p)$ implies
%$\Omc\models \tr{q}{p}$. The converse direction is analogous.
\end{proof}

Theorem~\ref{thm:combined1} establishes the main properties Algorithm~\ref{alg:perfectRef}.
As already mentioned, termination can be proved using
a similar argument as the one used for \perfectRef~\cite[Lemma 34]{dl-lite}.
The important point is about the correctness of the rewritings,
which we argue next, as part of our proof for Theorem~\ref{thm:combined}.

\Theoremqueryrewriting*
\begin{proof}
The lower bound follows from \NP-hardness of conjunctive query answering in
relational databases. We argue that the problem is in \NP.
Let $(\obdaSpecification,\database)$ be
an
 OBDA instance %, with $\obdaSpecification=(\Omc,\mapping,\Smc)$,
 and let $(q,p)$ be a  query.
By Lemma~\ref{lem:not-equal-marked},
there is
 $(\obdaSpecification_m,\dataInstance_m)$ and
   $(q_m,p_m)$
    such that
\begin{itemize}
\item $(\obdaSpecification,\dataInstance)\models (q,p)$ iff
$(\obdaSpecification_m,\dataInstance_m)\models (q_m,p_m)$;
\item for any two monomials $p_1,p_2$ appearing in $p_m$, it holds that $p_1$ and $p_2$
are mathematically distinct; and
\item  $|(\obdaSpecification_m, \database_m)|+|(q_m,p_m)|$ is polynomially bounded by
$|(\obdaSpecification,\database)|+|(q,p)|$.
\end{itemize}
By Theorem~\ref{thm:trans},
 for such polynomials $p_m\in\NPr$,
 $(\obdaSpecification_m,\database_m)\models (q_m,p_m)$ iff
there is $q'\in \tr{q_m}{p_m}$ s.t.\ $(\obdaSpecification_m,\database_m)\models q'$.
Then, to establish our upper bound we proceed as follows.

Given an   OBDA  instance $(\obdaSpecification,\database)$
and a  query $(q,p)$, we first check in $\LogSpace\subseteq \PTime$
satisfiability of  $(\obdaSpecification,\database)$~\cite{ArtEtAl1}.
If $(\obdaSpecification,\database)$ is unsatisfiable
then for all   queries $(q,p)$, we  have that
$(\obdaSpecification,\database)\models (q,p)$ holds trivially.
Then, assume $(\obdaSpecification,\database)$ is satisfiable.
We guess  an OBDA  instance $(\obdaSpecification_m,\database_m)$
 marked for $(\obdaSpecification,\database)$,
 a  query $(q_m,p_m)$, and
  $q'\in{\sf Tr}(q_m,p_m)$ such that
  $(\obdaSpecification,\database)\models (q,p)$
iff
$(\obdaSpecification_m,\database_m)\models q'$.
%   Moreover, the sizes of $(\obdaSpecification_m,\database_m)$ and
 %
 By construction of $q'$, $|q'|$ is polynomial in
 $|(q_m,p_m)|$, which in turn is polynomial in
  $|(\obdaSpecification,\database)|+|(q,p)|$ (Lemma~\ref{lem:not-equal-marked}).

  We now adapt the query rewriting
  algorithm \perfectRef~\cite{dl-lite} to decide whether
  $(\obdaSpecification_m,\database_m)\models q'$.
We %guess $q_m$, $p_m$ and   such $q'\in \tr{q_m}{p_m}$,
 guess a rewriting  $q^\ddagger$ of $q'$,
 and
 guess a sequence of pairs $(I,n)$ where $I$ is a positive inclusion
 and $n$ is an identifier for an atom position in a query.
 Each $(I,n)$   represents a transformation of  \perfectRef
 on a query.
There is a  non-deterministic version of \perfectRef which would follow this sequence of
transformations and return   $q^\ddagger$, with  $q'$ as input.
The sequence of transformations is of polynomial size, since every query returned
by \perfectRef can only be generated after a polynomial number of transformations
of the initial query.
% (i.e., after a polynomial number of executions of steps (a)
%and (b) of the algorithm).
%By definition of $\tr{q_m}{p_m}$, each $q'\in \tr{q_m}{p_m}$
%is of polynomial size in $|q_m|+|p_m|$ (which in turn are of polynomial size
%in $|q|+|p|$) and,
By definition of
\perfectRef, each $q^\ddagger \in \perfectRef(q',\Omc_\Tmc)$ is   polynomial
in $|q'|$, where $\obdaSpecification=(\Omc,\mapping,\Smc)$ and
$\Omc_\Tmc$ is the set of positive
inclusions in \Omc.

Membership in \NP follows from the fact that
%(1) satisfiability of annotated \DLLITE is in \LogSpace, as in \DLLITE~\cite{ArtEtAl1};
 we can check in polynomial time that:
 \begin{enumerate}
\item $q'\in \tr{q_m}{p_m}$;
\item $q^\ddagger \in \perfectRef(q',\Omc_\Tmc)$
(using the sequence of transformations);
\item
$p$ is the result of
replacing each occurrence of
$v\in\NV$  in $p_m $ by $\dagger(v)$,
where $\dagger$ is a function %$\dagger:\NV\rightarrow\NV$
that maps $v\in\NV$ occurring in $(\obdaSpecification_m,\database_m)$
to $\dagger(v)$ in $(\obdaSpecification,\database)$;
%(4) we can reconstruct $p$ using $p_m$ and the definition of $\Omc_m$;
\item $q$ is the result of  each occurrence of
$v\in\NV$  in $q_m $ by $\dagger(v)$,
 in addition to replacing
roles $R_{a,b}$ by $R$ in $q_m$, where $\dagger$ is as in Point (3);
and
%\item
\end{enumerate}
our modified  algorithm  is correct.
Correctness   is shown  as %with an argument similar to the one used for  $\perfectRef$
 in~\cite{dl-lite}, but here we change
the notion of   applicability of a positive inclusion $I$ to an atom $g$
and the definition of $gr(g,I)$ (Definition~\ref{def:atom}) to ensure that
each transformation respects the semantics of annotated \DLLITE.
%~ \ano{todo, labels attention}
%~ The complexity results (ii) and (iii) follow from the fact that:
%~ (1) satisfiability of annotated \DLLITE is in \LogSpace, as in \DLLITE~\cite{ArtEtAl1};
%~ (2) the query rewriting can be computed in time polynomial in the number
%~ of positive inclusions in an annotated \DLLITE ontology \Omc and in constant time
%~ in the number of assertions in \Omc; (3) (our modified version of) $\perfectRef$ is correct~\cite{dl-lite}
%~ (see Theorem~\ref{thm:combined});
%~ and (4) evaluation of a query over a database can be computed in \LogSpace
%~ w.r.t. the size of the database (follows from the fact that
%~ queries are a fragment of first-order logic and their size is
%~ fixed for the complexity results   in this theorem).
\end{proof}

\section{Proofs for Section~\ref{sec:idem}}

\Propositionexponential*
\begin{proof}
Consider the     OBDA instance with the ontology \Omc containing the axioms,
%with $1\le i<n$:
\[
  \begin{array}{l}
    (A\sqsubseteq B_1,x),  (A\sqsubseteq C_1,x), (B_n\sqsubseteq D,x),\\[1mm]
    (B_i\sqsubseteq B_{i+1},x_i),  (B_i\sqsubseteq C_{i+1},y_i),
    (C_i\sqsubseteq B_{i+1},x_i),\\[1mm]
    (C_i\sqsubseteq C_{i+1},y_i), (C_n,\sqsubseteq D,x),
  \end{array}
\]
for $1\le i<n$, and $\Mmc(\Dmc)=\{(A(a), p)\}$.
Consider the simple query $q=D(a)$. Every monomial $p=x\otimes\prod_{i=1}^n z_i$, where
each $z_i\in\{x_i,y_i\}$, is such that $\Omc\models(q,p)$ (and none other). Hence, the polynomial of the query
$q$ is formed by the sum of $2^n$ different monomials; that is, it is exponential on the size of the ontology.
\end{proof}

\PropIdempotent*
\begin{proof}
Let $N$ be a set with $n$ concept names such that $A,B\in N$. Let
$((\Omc,\Mmc,\Smc),\Dmc)$ be an   OBDA instance
with
%Define the ontology
$$\Omc=\{(C\sqsubseteq D,x_{CD})\mid C,D\in N\}
\text{, } \ \Mmc(\Dmc)=\{A(a),x_A\}.$$ Consider  the query $q=B(a)$. We have that
\Omc simulates a complete graph over the nodes in $N$. Every derivation of $B(a)$ represents a path
from $A$ to $B$ in this graph. Hence, if the provenance of $q$ could be expressed with polynomial space,
there would be a monotone Boolean formula representing all the paths from $A$ to $B$ in the complete
graph with $n$ nodes, contradicting the results from~\cite{KaWi88,KaWi90}.
\end{proof}

\TheoremComputingProv*
\begin{proof}
Recall that in Section~\ref{sec:idem} we consider idempotent semirings.
Assume  a standard BCQ $q$ and an OBDA instance $(\Pmc,\Dmc)$ is given as
input to Algorithm~\ref{alg:computeProv} and it returns $p\in\NPr$.
Then, it suffices to show that
every monomial $m\in\NM$ in $p$
 is such that $(\Pmc,\Dmc)\models (q,m)$ and,
conversely, if there is $m'\in\NM$ such that
$(\Pmc,\Dmc)\models (q,m')$ then there is a monomial in $p$ that is mathematically equal to $m'$
(by `monomial in a polynomial $p$' we
mean that the monomial is  one of the elements of the sum of monomials, given that
$p$ is in expanded form).

By definition of Algorihm~\ref{alg:computeProv},
if $m\in\NM$ is a monomial  in $p$ then
there is a match of $q'_{\vec{y}}$ where $q'$ is a rewriting
of $q^\star$ in $\Imc_{\Mmc(\Dmc)}$.
By Definition~\ref{def:atom2}, the last term of each atom
of the rewriting $q'$ is associated with the product
of $\star$ and the annotations of the inclusions
used to derive the atom.
We then replace $\star$ in each atom by the
corresponding annotation of the tuple in the match.
By soundness of the algorithm \perfectRef~\cite{dl-lite}
and the semantics of annotated interpretations,
 $(\Pmc,\Dmc)\models (q,m)$.
%, which is the monomial obtained by
%taking the product of the last term of each atom,

Conversely, if $(\Pmc,\Dmc)\models (q,m)$ then, by
the semantics of annotated interpretations,
$m$ corresponds to a derivation of $q$
using axioms of $(\Pmc,\Dmc)$. By completeness of
the algorithm \perfectRef~\cite{dl-lite}, there
is a query rewriting $q'$ of $q^\star$ in $PR$, following
  Definition~\ref{def:atom2},
and a match of $q'_{\vec{y}}$ in
$\Imc_{\Mmc(\Dmc)}$ such that  $m$ is the product
of the
annotations
resulting from replacing
$\star$ in the last term of each atom
of $q'$ by the corresponding annotation of
the tuple in the match.
By definition of Algorihm~\ref{alg:computeProv},
  $m$  is   added to the polynomial returned by the algorithm. % by Algorihm~\ref{alg:computeProv}.
%~ where,
%~ following
  %~ Definition~\ref{def:atom2},
 %~ we have that $q'$ carries in each atom the product
 %~ of $\star$ and the provenance information
 %~ of the inclusions used to derive the atom.
 %~ %By definition of Algorihm~\ref{alg:computeProv},
 %~ Then, the symbol $\star$ is replaced in each atom by the
%~ corresponding annotation of the tuple in a match of $q'_{\vec{y}}$
%~ such that the product of the annnotations results in the
   %~ monomial $m$  added to the polynomial
 %~ returned by Algorihm~\ref{alg:computeProv}.
\end{proof}

\section{Algorithm for Section~\ref{sec:evaluation}}

We first briefly recall the query answering algorithm used in \ontop,
and introduce the ProvSQL extension of PostgreSQL. Then we explain how
\ontoprov extends \ontop with the support for provenance.

\subsection{ProvSQL}
\label{sec:provsql}

The ProvSQL project is a PostgreSQL extension to add support for
(m-)semiring provenance. It supports semiring provenance, with or
without monus (m-semiring). Note that the monus operator is beyond the
semiring framework considered here. ProvSQL adds an extra
\textsl{provsql} column to each table in the database.
This column associates to each tuple a \emph{universally unique identifier} (UUID)
as a provenance token. Likewise, each answer is also associated to a UUID. ProSQL providdes several functions to intepret these tokens into several kinds of provenance information. The next example shows how ProvSQL works.
%the to return the circuit of the provenance rooted at \texttt{token}:
%\texttt{view\_circuit()} returns a PDF visualization
%of the subcircuit, and
%\texttt{where\_provenance()} returns a textual representation of
%the where-provenance.

\begin{example}\label{ex:provsql}
Consider the following two tables

\begin{center}
  \footnotesize
    \begin{tabular}{lllll}
      \multicolumn{5}{c}{EMP} \\
      \toprule
      \underline{EMPNO} & ENAME & JOB & DEPTNO & \textsl{provsql} \\ \midrule
      7367 & SMITH & CLERK & 10 & t11 \\
      9527 & JOHN  & HR & 10 & t12 \\
      4839 & MARY  & PROGR & 10 & t13 \\
      4839 & RALPH  & SYSADM & 10 & t14 \\
      \bottomrule
    \end{tabular}
    \\~\\~\\
    \begin{tabular}{llll}
      \multicolumn{4}{c}{DEPT} \\
      \toprule
      \underline{DEPTNO} & DNAME & LOC & provsql\\ \midrule
      10 & APPSERVER & NEW YORK & t21\\
      \bottomrule
    \end{tabular}
\end{center}
By calling the \texttt{add\_provenance} functions, ProvSQL adds a column \texttt{provsql} to both tables, add generates a unique provenance token for each tuple.

Consider the SQL query:

\begin{lstlisting}[basicstyle=\footnotesize\ttfamily]
SELECT e.ENAME, d.DNAME
FROM EMP e, DPT d
WHERE e.DEPTNO = d.DEPTNO
   AND JOB='PROGR'
\end{lstlisting}
To this query ProvSQL returns the following result

\begin{center}
  \begin{tabular}{lll}
    \toprule
    ENAME & DNAME & provsql\\ \midrule
    Mary & APPSERVER & t31\\
    \bottomrule
  \end{tabular}
\end{center}

The additional column \textsl{provsql} is added automatically by ProvSQL, and it contains a fresh token which encodes the provenance information for the tuple in the result. Such token can be then fed to ProvSQL functions, such as the \texttt{where\_provenance} function, to decode the desired provenance information. In our example, by using \texttt{where\_provenance('t31')} we would get

\begin{lstlisting}[basicstyle=\footnotesize\ttfamily]
{[EMP:t13:2],[DEPT:t21:2]}
\end{lstlisting}
saying that the result tuple has been derived from the second attributes of the two tuples t13 and t21 in the tables EMP and DEPT.
\end{example}

\subsection{Algorithm of \ontop}

\begin{algorithm}[tb]
  \caption{Ontop}\label{al:ontop}
\textbf{Input:} a query $q$, an OBDA instance $I = ((\Omc,\Mmc,\Smc),\database)$\\
\textbf{Output:} the answer to $q$ w.r.t.\ $I$

\label{alg:ontop-workflow}
\begin{algorithmic}[1]
  \STATE \text{// offline} %
  \STATE $\O' = \textit{classify}(\O)$ %
  \STATE $\M^{sat} = \textit{saturate}(\M,\O')$ %
  \STATE $\M^{sat} = \textit{optimizeM}(\M^{sat})$ %
  \STATE \text{// online} %
  \STATE $q' = \textit{unfold}(q, \M^{sat})$ %
  \STATE $Q = \textit{optimizeQ}(q',S)$ %
  \STATE \textbf{return} $\textit{eval}(Q,\D)$ %
\end{algorithmic}
\end{algorithm}

In \ontop, tree-witness rewriting is switched off by default. This because it is recognized that queries which trigger tree-witness rewriting are mostly of theoretical relevance, and extremely rare in practice (actually we are not aware of any real-world scenario in which such queries are utilized). A reason for this is that whereas in CQs it is natural to write queries with existentially quantified variables (non-answer variables are always existentially quantified), in SPARQL (under the OWL~2~QL entailment regime) doing so is more difficult.

\ontop is an OBDA system that operates over OWL~2~QL
ontologies~\cite{W3Crec-OWL2-Profiles} and
SPARQL~\cite{W3Crec-SPARQL-1.1-query} queries.  The standard semantics for such
setting is the OWL~2~QL entailment regime~\cite{W3Crec-SPARQL-1.1-entailment},
which slightly diverges from the one adopted in the DL
context~\cite{KRRXZ14}. For performance reasons, \ontop does not rely on
\perfectRef, but instead adopts a mixed strategy based on \emph{mapping
 saturation} and \emph{tree-witness rewriting}~\cite{CCKK*17}.  Mapping
saturation simulates the saturation of the ABox with respect to basic concept
and role inclusions, whereas tree-witness rewriting rewrites the query so as to
take into account those axioms with an existential on the right-hand side.

%% Footnotes do not seem to work!
% \footnote{\ontop currently does not support SPARQL queries that contain the
% OWL value restriction construct (which is allowed in OWL~2~QL concept
% expressions), through which one would be able to simulate existentially
% quantified variables in conjunctive queries, provided such variables do not
% appear in a cycle in the query.}.

In \ontop, tree-witness rewriting is switched off by default.  Indeed, for
compliance with the OWL~2~QL entailment regime, also non-answer variables in a
SPARQL query have to match known individuals (and not existentially implied
ones).  Therefore, SPARQL queries for which applying tree-witness rewriting may
actually result in additional answers are rather unnatural (they are obtained by
constructing a complex concept expression that makes use of value restriction
within the query itself), and it is commonly recognized that they are rare in
practical scenarios (actually we are not aware of any real-world scenario in
which such queries are utilized).  Note that this differs from the semantics of
CQs, where every non-answer variable can match also existentially implied
individuals.  Therefore it is more natural to write CQs for which applying
tree-witness rewriting has a concrete impact on query answering.

We have adapted the mapping saturation process to our setting with provenance,
so as to support query answering with provenance in relevant practical
scenarios, and be able to measure the performance in them. A full
implementation that also adapts the tree-witness rewriting is nevertheless
important and will be part of future work, but it is not relevant for the
(real-world) tests in this work.
The \ontop algorithm, limited to the mapping saturation approach, is outlined
in Algorithm~\ref{alg:ontop-workflow}. The algorithm takes as inputs an OBDA
instance $I$ and a SPARQL query $q$, and returns the answers to $q$ w.r.t.\
$I$.  The workflow can be divided into an offline stage and an online stage.
During the offline stage (i.e., system start-up), \ontop first
classifies the ontology $\O$. The result of the
classification is a complete hierarchy of classes and properties,
which is stored in-memory as a directed acyclic graph.  Then it
compiles the classified ontology into the input mapping $\M$, thus
obtaining the saturated mapping $\M^{sat}$, also known as the T-mapping.

During the online stage (i.e., query execution), \ontop transforms the
input SPARQL query $q$ into an SQL query $q'$ by exploiting the saturated-mapping $\M^{sat}$, and then produces an optimized SQL query $Q$ by exploiting the
database integrity constraints $\S$. Finally, $Q$ is evaluated over the database
instance $\D$.

Next example shows the steps from Algorithm~\ref{al:ontop}.

\begin{example}\label{ex:ontop}
% ONTOLOGY
Consider the setting from Example~\ref{ex:provsql}. Consider the following ontology stating that programmers are employees.

\begin{lstlisting}[basicstyle=\scriptsize\ttfamily,breaklines=true,frame=tb]
ax1 SubClassOf(:Programmer, :Employee).
\end{lstlisting}

Consider the following mapping assertions:

\begin{lstlisting}
MapID:  m1
Target: triple(:emp/{EMPNO},rdf:type,:Employee).
        triple(:emp/{EMPNO},ex:name,{ENAME}).
        triple(:emp/{EMPNO},ex:dept,:dept/{DEPTNO}).
Source: SELECT * FROM EMP

MapID:  m2
Target: triple(:dept/{DEPTNO},rdf:type,:Department).
        triple(:dept/{DEPTNO},ex:name,{DNAME}).
        triple(:dept/{DEPTNO},ex:loc, {LOC}).
Source: SELECT * FROM DEPT

MapID: m3
Target: triple(:emp/{EMPNO},rdf:type,:Programmer).
Source: SELECT * FROM EMP WHERE JOB='PROGR'
\end{lstlisting}

After ontology classification and mapping saturation, the saturated set of mappings will contain the original mappings plus the following mapping:

\begin{lstlisting}
MapID: m3_ax1
Target: triple(:emp/{EMPNO},rdf:type,:Employee).
Source: SELECT * FROM EMP WHERE JOB='PROGR'
\end{lstlisting}

This mapping is derived from mapping \texttt{m3} and axiom \texttt{ax1}. Since the SQL query in \texttt{m3\_ax1} is contained in the SQL query in \texttt{m1}, the mapping is in fact redundant and gets removed in the \texttt{optimizeM} step. Therefore, the final saturated set of mappings will coincide with the original set of mappings.

Consider the following SPARQL query relating employees to departments they work in:

\begin{lstlisting}
SELECT ?eName ?dName WHERE {
   ?e rdf:type :Employee;
      ex:name  ?eName;
      ex:dept  ?d .
   ?d ex:name dName.
}
\end{lstlisting}
Such query has the following algebra tree:

\begin{lstlisting}
SELECT ?eName ?dName
   JOIN
      triple(?e, rdf:type, :Employee).
      triple(?e, ex:name, ?eName).
      triple(?e, ex:dept, ?d).
      triple(?d, ex:name, ?dName).
\end{lstlisting}

In the unfold step, each \texttt{triple} in the algebra tree is replaced by the corresponding source part in the mapping. The query $q'$ look like:

\begin{lstlisting}
SELECT e1.ENAME as eName, d.DNAME as dName,
FROM EMP e1, EMP e2, EMP e3, DPT d
WHERE e1.EMPNO=e2.EMPNO AND
      e1.EMPNO=e3.EMPNO AND
      e1.DEPTNO = d.DEPTNO AND
      e1.JOB='PROGR'
\end{lstlisting}

In the \texttt{optimizeQ} step, the redundant self-joins are removed and the final query $Q$ will be:

\begin{lstlisting}[basicstyle=\scriptsize\ttfamily,breaklines=true,frame=tb]
SELECT e.ENAME as eName, d.DNAME as dName,
FROM EMP e, DPT d
WHERE e.DEPTNO = d.DEPTNO
   AND JOB='PROGR'
\end{lstlisting}
\end{example}

\subsection{The \ontoprov System}

\ontoprov accepts a BGP query, possibly with a FILTER condition, and
returns the answer of the query together with the provenance for each
answer tuple.

Algorithm~\ref{al:ontoprov} outlines the \ontoprov approach. The main workflow is the same, but now each step has to deal also with the provenance information.

We explain the algorithm by means of an example.

\begin{algorithm}[tb]
  \caption{OntoProv}\label{al:ontoprov}
  \textbf{Input:} a query $q$, an OBDA instance $I = ((\Omc,\Mmc,\Smc),\database)$\\
  \textbf{Output:} the answer to $q$ w.r.t.\ $I$
  \label{alg:ontoprov-workflow}
\begin{algorithmic}[1]
  \STATE \text{// offline} %
  \STATE $\O' = \textit{classifyAndStorePaths}(\O)$ %
  \STATE $\M_{prov} = addMappingsProvenance(\M)$ %
  \STATE $\M' = \M \cup \M_{prov}$ %
  \STATE $\M^{sat} = \textit{saturate}(\M',\O')$ %
  \STATE $\M^{sat} = \textit{optimizeM}(\M^{sat})$ %
  \STATE \text{// online} %
  \STATE $q^{\M^{sat}} = \textit{unfold}(q, \M^{sat})$ %
  \STATE $Q = \textit{optimizeQ}(q^{\M^{sat}},S)$ %
  \STATE \textbf{return} $\textit{eval}(Q,\D)$ %
\end{algorithmic}
\end{algorithm}

\begin{example}
Consider the scenario from Example~\ref{ex:ontop}. The function \texttt{classifyAndStorePaths} classifies the ontology and stores, for each (basic) concept and role in the ontology, the paths to their descendants. For our example, it stores the path:

\begin{lstlisting}
p[:Employee,:Engineer]
\end{lstlisting}

In \texttt{addMappingsProvenance}, the mappings are used to generate new mappings encoding the provenance information (i.e., mappings IDs and provenance returned by ProvSQL) in their target parts (which are now quadruples). For our example, these mappings are:

\begin{lstlisting}
MapID:  m1quad
Target: quad(:emp/{EMPNO},rdf:type,:Employee,
               :prov/mkey-m1/dkey-{provsql}).
        quad(:emp/{EMPNO},ex:name,{ENAME},
               :prov/mkey-m1/dkey-{provsql}).
        quad(:emp/{EMPNO},ex:dept,:dept/{DEPTNO},
               :prov/mkey-m1/dkey-{provsql}).
Source: SELECT * FROM EMP

MapID:  m2quad
Target: quad(:dept/{DEPTNO},rdf:type,:Department,
                :prov/mkey-m2/dkey-{provsql}).
        quad(:dept/{DEPTNO},ex:name,{DNAME},
                :prov/mkey-m2/dkey-{provsql}).
        quad(:dept/{DEPTNO},ex:loc, {LOC},
                :prov/mkey-m2/dkey-{provsql}).
Source: SELECT * FROM DEPT

MapID: m3quad
Target: quad(:emp/{EMPNO},rdf:type,:Programmer,
                  :prov/mkey-m3/dkey-{provsql}).
Source: SELECT * FROM EMP WHERE JOB='PROGR'
\end{lstlisting}

The saturated set of mappings is derived as in Example~\ref{ex:ontop}. Hence, it will contain the following mapping assertion:

\begin{lstlisting}
MapID: m3_ax1quad
Target: quad(:emp/{EMPNO},rdf:type,:Employee,
               :prov/mkey-m3/okey-p[:Programmer,:Employee]/dkey-{provsql}).
Source: SELECT * FROM EMP WHERE JOB='PROGR'
\end{lstlisting}

Such mapping encodes also the ontology axioms that have been used to derive it. In particular, it encodes the path
\begin{center}
  \footnotesize
  \verb|okey-p[:Programmer,:Employee]|
\end{center}

Observe that, although the SQL query in \texttt{m3\_ax1quad} is contained in the SQL query in \texttt{m1quad}, the target parts of such mappings do not coincide. Therefore, the added mapping is \emph{not} redundant, and it will not be removed by the function \texttt{optimizeM}.

The input query gets rewritten into an equivalent query containing \emph{named graphs patterns}:

\begin{lstlisting}
SELECT ?eName ?dName ?p1 ?p2 ?p3 ?p4 WHERE {
   GRAPH ?p1 {?e rdf:type :Employee.}
   GRAPH ?p2 {?e ex:name  ?eName.}
   GRAPH ?p3 {?e ex:dept  ?d .}
   GRAPH ?p4 {?d ex:name dName.}
}
\end{lstlisting}

The algebra tree for such query is:

\begin{lstlisting}
  SELECT ?eName ?dName ?p1 ?p2 ?p3 ?p4
   JOIN
      quad(?e, rdf:type, :Employee, ?p1).
      quad(?e, ex:name, ?eName, ?p2).
      quad(?e, ex:dept, ?d, ?p3).
      quad(?d, ex:name, ?dName, ?p4).
\end{lstlisting}

The unfolding and optimizations proceed in the same way as for Example~\ref{ex:ontop}. Due to the presence of the mapping \texttt{m3\_ax1quad}, the final SQL query will contain a union:

\begin{lstlisting}
SELECT e.ENAME as eName, d.DNAME as dName,
       ':prov/mkey-m1/dkey-' || e.provsql as p1,
       ':prov/mkey-m2/dkey-' || e.provsql as p2
FROM EMP e, DPT d
WHERE e.DEPTNO = d.DEPTNO
UNION
SELECT e.ENAME as eName, d.DNAME as dName,
       ':prov/mkey-m3/okey-p[:Programmer,:Employee]/dkey-' || e.provsql as p1,
       ':prov/mkey-m2/dkey-' || e.provsql as p2
FROM EMP e, DPT d
WHERE e.DEPTNO = d.DEPTNO
   AND JOB='PROGR'
\end{lstlisting}
Such union keeps track of the fact that employees can either be derived by the mapping $m1$ alone, or by exploiting $m3$ together with the ontology axiom.
\end{example}

\fi

\end{document}